\documentclass[11pt,twoside,titlepage]{report}
\usepackage[lmargin=2cm,rmargin=2cm,tmargin=1.5cm,bmargin=1.5cm]{geometry}

\usepackage{COMSOC_macros}

\begin{document}

\title{
Lecture Notes on \\
{\bf Voting Theory}}

\author{{\bf Davide Grossi}
\bigskip \\
Bernoulli Institute for Maths, CS and AI \\
{\em University of Groningen} 
\medskip \\
Amsterdam Center for Law and Economics \\
Institute for Logic, Language and Computation \\
{\em University of Amsterdam} 
\bigskip \\
\href{http://www.davidegrossi.me}{www.davidegrossi.me} 
}
\bigskip 
\date{\textcopyright Davide Grossi 2021}

\maketitle

%%%%%%%%%%%%%%%%%%%%%%%%%%%%%%%%%%%%%%%%%%%%%%%%%%%%%

% Table of Contents

\thispagestyle{empty}
\tableofcontents
\newpage

%%%%%%%%%%%%%%%%%%%%%%%%%%%%%%%%%%%%%%%%%%%%%%%%%%%%%%%

%%%%%%%%%%%%%%%%%%%%%%%%%%%%%%%%%%%%%%%%%%%%%%%%%%%%%%%%%%%%%%%

%\begin{abstract}

\thispagestyle{empty}

\vspace*{5cm}

\subsection*{Foreword}

These lecture notes have been developed for the course \textit{Computational Social Choice} of the Artificial Intelligence MSc programme at the University of Groningen (academic years 2019/20 and 2020/21). They cover mathematical and algorithmic aspects of voting theory. 

%The course has been inspired by similar courses run by colleagues at the University of Amsterdam (Ulle Endriss) and Technion (Reshef Meir).

\medskip

The author wishes to thank the students of the above mentioned courses, and Yuzhe Zhang, for the many corrections and suggestions provided for these notes.

%\end{abstract}

\pagestyle{fancy}
\setcounter{chapter}{0}
\setcounter{page}{1}

%%%%%%%%%%%%%%%%%%%%%%%%%%%%%%%%%%%%%%%%%%%%%%%%%%%%%%%%%%%%%%%

\chapter{Choosing One Out of Two}

%%%%%%%%%%%%%%%%%%%%%%%%%%%%%%%%%%%%%%%%%%%%%%%%%%%%%%%%%%%%%%%

%%%%%%%%%%%%%%%%%%%%%%%%%%%%%%%%%%%%%%%%%%%%%%%%%%%%%%%%%%%%%%%

This chapter introduces the basic ideas and definitions of standard voting theory and applies them to the context in which a group has to choose one out of two alternatives. In this context we prove two fundamental theorems of social choice over two alternatives.

%%%%%%%%%%%%%%%%%%

\section{Preliminaries}

\subsection{Key definitions}

A social choice problem arises whenever a group needs to take a collective choice, among two or more options, based on the different preferences or opinions that the members of the group may have about those options. We can make this social choice context precise as follows:
\begin{itemize}
\item $\N$ is a finite set of $n$ {\em voters}, or agents (so, $|N| = n$). We specifically take $\N$ to be a set of natural numbers: $\N = \set{1, 2, \ldots, n}$.
\item $\A$ is a finite set of {\em alternatives}, or options, of size $m$ (so, $|A| = m$). We assume that $m \geq 2$. Alternatives are denoted by the initial letters of the alphabet: $a$, $b$, $c$, etc.
\end{itemize}

Each voter expresses a preference, or {\em ballot} $\p_i$. Such ballots are taken to be  linear orders\footnote{That is, a binary relations that are transitive, total (or complete), and antisymmetric (and therefore reflexive).} over $\A$: $x \p_i y$ stands for ``$x$ is strictly preferred to $y$ by $i$, or $x = y$''. The set of all linear orders over $\A$ is denoted $\LO(\A)$. So, for all $i \in \N$, $\p_i \in \LO(\A)$. The asymmetric (and thus irreflexive) part of each $\p_i$ is denoted $\succ_i$.\footnote{Note that $\p$ and $\succ$ are identical except for the fact that the first is reflexive and the second is irreflexive.} 

Collecting all these ballots together defines a so-called preference (or ballot) profile $\P = \tuple{\p_1, \p_2, \ldots, \p_n}$. That is, a tuple collecting all ballots of the voters in $\N$ about the alternatives in $\A$. Such a tuple is called a {\em profile}. The set of all profiles for the voters in $\N$ is denoted $\LO(\A)^n$.\footnote{
We can also conveniently think of a profile $\P$ as a matrix
\[
\P = 
\begin{bmatrix}
    x_{11} & x_{12} & x_{13} & \dots  & x_{1m} \\
    x_{21} & x_{22} & x_{23} & \dots  & x_{2m} \\
    \vdots & \vdots & \vdots & \ddots & \vdots \\
    x_{d1} & x_{d2} & x_{d3} & \dots  & x_{nm}
\end{bmatrix}
\]
where each entry $x_{ij}$ denotes the alternative in the $j^\mathit{th}$ position in the linear order of voter $i$. Using matrix notation we can write $\P_{ij}$ for such an entry. 
Similarly $\P_i$ denotes the vector representing the linear order $\p_i$ of voter $i$.
} 

\medskip

So the problem of social choice, loosely defined, is the problem of finding a suitable subset of alternatives that makes as many voters as `happy' as possible, given the preference information revealed by the ballots in a profile.

\begin{example}[A preference profile] \label{ex:profile}
Let $\N = \set{1, 2 , 3}$ and $A = \set{a, b}$. A possible preference profile is the triple $\tuple{\p_1, \p_2, \p_3}$ where: $\p_1 =  ab$, $\p_2 =  ab$ and $\p_3 =  ba$.\footnote{We denote linear orders simply as sequences, or vectors, of alternatives from the most preferred (left) to the least preferred (right).} This profile can be represented in tabular form as follows:
\begin{center}
\begin{tabular}{l|cc}
$1$ & $a$ & $b$ \\
$2$ & $a$ & $b$ \\
$3$ & $b$ & $a$
\end{tabular}
\end{center}
Here each row represents a voter, and the linear order representing to the voter's ballot is rendered from the most preferred (left) to the least preferred (right).
\end{example}

\begin{definition}
Let $\tuple{\N, \A}$ be given. A {\em social choice function} or SCF for $\tuple{\N, \A}$ is a function
\begin{align}
f: \LO(\A)^n \to 2^A \setminus \emptyset. \label{eq:SCF}
\end{align}
\end{definition}
That is, for any profile $\P$, an SCF outputs a non-empty set of alternatives. These are the alternatives constituting the `social choice' made by the group $\N$ via function $f$, given the expressed preferences in $\P$. If the size of the output of an SCF is larger than $1$ then the alternatives it contains can be thought of as `tied' choices. 

We will refer to concrete social choice functions as {\em voting rules}. The possibly best-known voting rule selects a single alternative as the social choice based on the number of first positions that each alternative gets in the ballots of the input voting profile:

\begin{VR}[Plurality] \label{vr:plurality}
The plurality rule is the SCF defined as follows. For any profile $\P \in \LO(\A)^n$:
\begin{align*}
\Plr(\P) & =  \set{x \in \A ~\middle|~ \nexists y \in \A, \size{\set{i \in \N \mid \max_{\p_i}(\A) = y}} > \size{\set{i \in \N \mid \max_{\p_i}(\A) = x }}}
\end{align*}
We will also refer to $\size{\set{i \in \N \mid \max_{\p_i}(\A) = x}}$ as the {\em plurality score} of alternative $x$ in profile $\P$ and we will sometimes denote it by $\Plr(\P)(x)$.
\end{VR}
%\begin{description}
%\item[Plurality rule] An SCF $f$ is the plurality rule (denoted $\Plr$) if, for any $x \in \A$:
%\begin{align}
%x \in f(\P) & \iff \nexists y \in \A, \size{\set{i \in \N \mid \max_{\p_i}(\A) = y}} > \size{\set{i \in \N \mid \max_{\p_i}(\A) = x }} 
%\end{align}
%\end{description}
So according to plurality, the social choice for a profile contains all alternatives for which there is no other alternative that occurs as top preference more often in the profile. Notice that this allows for ties in the social choice.

\begin{example}[Plurality] 
Applying the plurality rule to the profile $\tuple{\p_1, \p_2, \p_3}$ in Example \ref{ex:profile} we obtain $\Plr(\tuple{\p_1, \p_2, \p_3}) = \set{a}$. The single winner is therefore alternative $a$. 
\end{example}

A useful benchmark for the further exploration of the space of voting rules are dictatorships.
\begin{VR}[Dictatorship] \label{vr:dictatorship}
The dictatorship of a given $i \in \N$ is the SCF defined as follows. For any profile $\P \in \LO(\A)^n$:
\[
\Dr_i(\P) = \set{\max_{\p_i}(\A)}.
\]
\end{VR}

%%%%%%%%%%%%%%%%%%%%%%

\subsection{Basic axioms}

Asking what is the best way to aggregate the preferences of a group in order to reach a social choice amounts to asking whether we can identify a `best' social choice function $f$.
For each $\tuple{\N, \A}$ there are many social choice functions: $(2^m-1)^{|\LO(\A)^n|}$. So even just for the social choice context of Example \ref{ex:profile} there are $(2^3-1)^{2^3}= 2562890625$ such functions, among which plurality and dictatorships are just some of them. 

A way to get a hold of this space is by imposing axioms on the SCF $f$ we would like to identify. Axioms are meant to capture properties of SCFs which are considered desirable. Ultimately one would like to identify a set of such desirable properties that, all together, uniquely determine ({\em characterize}) $f$. Each of the axioms reported below captures a property of SCFs that may be considered desirable, in at least some contexts of collective decision-making. Some may be considered more intuitive than others, but no axiom can really be considered unquestionable.

\medskip

The first group of axioms concerns some basic requirements roughly inspired by an idea of `democratic voting' in which all voters and all alternatives should be treated equally.
\begin{definition}[Equal treatment axioms] \label{def:equal}
Let $\tuple{\N, \A}$ be given. An SCF $f$ is:
\begin{description}
\item[Anonymous] iff for every permutation $\pi: \N \to \N$ and $\P \in \LO(\A)^n$, 
\[
f(\underbrace{\tuple{\p_1, \p_2, \ldots, \p_n}}_\P) = f\left(\tuple{\p_{\pi(1)}, \p_{\pi(2)}, \ldots, \p_{\pi(n)}}\right).
\]
{\em Intuitively}, $f$ treats all voters equally. 

\item[Non-dictatorial] iff there exists no $i \in \N$ s.t. $\forall \P \in \LO(\A)^n$, $f(\P) = \set{\max_{\p_i}(\A)}$. Otherwise $f$ is said to be dictatorial with $i$ being {\em the dictator}.

{\em Intuitively}, there exists no voter that can unilaterally determine the social choice. 

\item[Neutral] if for every permutation $\rho: \A \to \A$ and $\P \in \LO(\A)^n$, 
\[
\rho(f(\underbrace{\tuple{\p_1, \p_2, \ldots, \p_n}}_\P)) =  f(\tuple{\rho(\p_1), \rho(\p_2), \ldots, \rho_i(\succeq_n)})
\] 
where for $\succeq_i = x_1x_2 \ldots x_n$, $\rho(\p_i) = \rho(x_1) \rho(x_2) \ldots \rho(x_m)$.

{\em Intuitively}, $f$ treats all alternatives equally. 
\end{description}
\end{definition}

The next group of axioms concerns the ability of an SCF to respond appropriately to growing support for alternatives. They rule out odd SCFs such as 'choose $x$ whenever the number of voters supporting $x$ is odd', or SCFs allowing for too many ties (e.g., `choose all alternatives supported by at least one third of the voters'). Before introducing the axioms we define one extra piece of notation. Given a profile $\P$ and two distinct alternatives $x$ and $y$, $\N_\P^{xy} = \set{i \in \N \mid x \pp_i y}$ denotes the set of voters that strictly prefer $x$ to $y$ in $\P$. Throughout these notes will also write $\supp^{xy}_\P$ instead of $\size{\N_\P^{xy}}$ to denote the number of voters that strictly prefer $x$ to $y$ in $\P$.

\begin{definition}[Responsiveness axioms] \label{def:resp_axioms}
Let $\tuple{\N, \A}$ be given. An SCF $f$ is:
\begin{description}

\item[Pareto] iff for all $\P \in \LO(\A)^n$ and for all $x \in \A$, if there exists $y \in \A$ s.t. for all $i \in \N$ $y \succ_i x$, then $x \not\in f(\P)$.

%there exists no $a \in \A$ s.t. $a \in f(\P)$ and there exists $b \in \A$ s.t. for all $i \in \N$, $a \p_i b$.

{\em Intuitively}, the social choice cannot contain dominated alternatives, that is, alternatives for which the unanimity of voters prefers a different alternative.

\item[Unanimous] iff for all $\P \in \LO(\A)^n$ and $x \in \A$, if for all $i \in \N$ $\max_{\p_i}(\A) = x$, then $f(\P) = \set{x}$.

{\em Intuitively}, whenever an alternative is the top choice of all voters, that alternative is the unique social choice.

\item[Monotonic]  iff for all $\P \in \LO(\A)^n$ and $x \in \A$, if $x \in f(\P)$ then $x \in f(\P')$, where:
\begin{itemize}
\item $\P' \in \LO(\A)^n$
\item $\N^{xy}_\P \subseteq \N^{xy}_{\P'}$, for all $y \in \A \setminus \set{x}$
\item $\N^{yz}_\P = \N^{yz}_{\P'}$, for all $y, z \in \A \setminus \set{x}$
\end{itemize}

{\em Intuitively}, if an alternative is a social choice it remains so whenever a voter decides to rank it higher in her ballot. 

\item[Positively responsive] iff for all $x \in \A$ and $\P \in \LO(\A)^n$, if $x \in f(\P)$ then $f(\P') = \set{x}$ where:
\begin{itemize}
\item $\P \neq \P' \in \LO(\A)^n$
\item $\N^{xy}_\P \subseteq \N^{xy}_{\P'}$, for all $y \in \A \setminus \set{x}$
\item $\N^{yz}_\P = \N^{yz}_{\P'}$, for all $y, z \in \A \setminus \set{x}$
\end{itemize}
% $\P' \in \LO(\A)$ and all $y, z \in \A \setminus \set{x}$ s.t. $\N^{xy}_\P \subseteq \N^{xy}_{\P'}$ and $\N^{yz}_\P = \N^{yz}_{\P'}$, then 
%$f(\P') = \set{x}$.

{\em Intuitively}, if an alternative is selected as a social choice in a profile, it becomes the only social choice those profiles in which some voter ranks $x$ higher.

\end{description}
\end{definition}

The axioms in this last group concern the ability of an SCF to output a single social choice, therefore ruling out ties, and its ability to effectively consider the whole space of possible preferences in the voters' population.
\begin{definition}[Imposition and resoluteness axioms]
Let $\tuple{\N, \A}$ be given. An SCF $f$ is:
\begin{description}

\item[Non-imposed] iff for all $x \in \A$ there exists $\P \in \LO(\A)^n$ s.t. $f(\P) = \set{x}$. It is said to be {\em imposed} otherwise.

{\em Intuitively}, every alternative is selected as unique social choice for at least one profile.

\item[Resolute] iff for all $\P \in \LO(\A)^n$, $f(\P)$ is a singleton.

{\em Intuitively}, the social choice is unique.

\end{description}
An SCF is imposed whenever there exists an alternative that can never be the social choice. It is irresolute when it admits ties.

\end{definition}

Some immediate consequences can be drawn from the above definitions. If $f$ is anonymous, then for all profiles $\P$ and $\P'$ such that, for all distinct alternatives $x$ and $y$ it holds that $\supp^{xy}_\P = \supp^{xy}_{\P'}$, we have that $f(\P) = f(\P')$. That is, for anonymous SCFs the only information that matters in a profile is the size of support for each alternative against each other alternative. As a direct consequence dictatorships are not anonymous. We state here some further observations:

\begin{fact} \label{fact:axioms}
Let $\tuple{\N, \A}$ be given, and let $f$ ba an SCF.
\begin{enumerate}[a)]
\item If $f$ is non-imposed and monotonic, then it is unanimous; 
\item If $f$ is is Pareto, then it is unanimous.
\item If $f$ is unanimous, then it is non-imposed.
 \end{enumerate}
\end{fact}
\begin{proof}
Exercise \ref{ex:axioms}
\end{proof}

%%%%%%%%%%%%%%%%%%

\section{Plurality is the best \ldots when $m = 2$}

We discuss now the plurality rule in the context of social choice problems where $m = 2$. Observe immediately that in such cases the plurality rule (Rule \ref{vr:plurality}) simplifies as follows, for all $\P \in \LO(\A)^n$:
\begin{align}
\Plr(\P) & =  \set{x \in \A ~\middle|~ \nexists y \in \A, \size{\set{i \in \N \mid \max_{\p_i}(\A) = y}} > \size{\set{i \in \N \mid \max_{\p_i}(\A) = x }}} \nonumber \\
 & =  \set{x \in \A ~\middle|~ \nexists y \in \A,  \supp^{yx}_\P > \supp^{xy}_\P} \nonumber \\
 & = 
 \left\{
 \begin{array}{ll}
 \set{x} & \mathit{if}~\supp^{xy}_\P \geq \lceil\frac{n+1}{2}\rceil   \\
 \set{y} & \mathit{if}~\supp^{yx}_\P \geq  \lceil\frac{n+1}{2}\rceil  \\
  \A & \mathit{otherwise}
 \end{array}
 \right.
 \label{eq:simple}
\end{align}
where $y$ is the second element of $\A$. The rule defined by \eqref{eq:simple} is also known as {\em simple majority}.

In this section we will also be working with the following generalization of the simple majority rule, under the assumption that $m=2$.
\begin{VR}[Quota] \label{vr:quota}
Let $\A$ be such that $m = 2$. The quota rule $\Qr_q$, where $q$ (the quota) is an integer such that $1 \leq q \leq n+1$, is an SCF defined as follows. For any profile $\P \in \LO(\A)^n$:
\begin{align*}
\Qr_q(\P) & = 
\left\{
 \begin{array}{ll}
 \set{x \in \A ~\middle|~ \size{\set{i \in \N \mid \max_{\p_i}(\A) = x}} \geq q} & \mbox{{\em if such set is non-empty}} \\
 \A & \mbox{{\em otherwise}}
 \end{array}
 \right.
\end{align*}
\end{VR}
Clearly, simple majority \eqref{eq:simple} is the quota rule where $q = \lceil \frac{n+1}{2} \rceil$. 
More generally, quota rules select an alternative as social choice whenever the support for that alternative reaches a given quota. Ties can arise when the quota are below a majority quota, or when the quota is impossible to be met $(q = n+1)$. 

\begin{fact}
Let $\A$ be such that $m = 2$. If $\frac{n}{2} < q \leq n + 1$ then:
\begin{align}
\Qr_q(\P) & =  
\left\{
 \begin{array}{ll}
 \set{x} & \mathit{if}~\supp^{xy}_\P \geq q  \\
 \set{y} & \mathit{if}~\supp^{yx}_\P \geq q \\
 \A & \mathit{otherwise}
 \end{array}
\right.
\end{align}
\end{fact}
\begin{proof}
It suffices to observe that by the definition of Rule \ref{vr:quota}, for every profile $\P$ and alternatives $x,y \in \A$, if $\frac{n}{2} < q \leq n$ then $x \in \Qr_q(\P)$ implies $y \not\in \Qr_q(\P)$.
\end{proof}

\subsection{Axiomatic characterizations of plurality when $m=2$}

The following three theorems provide an axiomatic justification for the use of simple majority voting in social choice contexts involving only two options.

\begin{theorem}[May's theorem \cite{may52set}] \label{th:may}
Let $\tuple{\N, \A}$ be given such that $n$ is odd and $m=2$. The plurality rule is the only SCF that is resolute, anonymous, neutral and monotonic.
\end{theorem}

\begin{proof}
\LtoR Assuming that $f$ is plurality, we need to show that $f$ is therefore resolute, anonymous, neutral and monotonic. See Exercise \ref{ex:may} \RtoL Let $f$ be resolute, anonymous, neutral and monotonic. We need to show that $f = \Plr$. Observe first of all that since $m=2$ there are only two possible ballots: $xy$ or $yx$. Furthermore, since $f$ is resolute there are only two possible outcomes: $\set{x}$ or $\set{y}$. Let us proceed towards a contradiction, and assume that $f \neq \Plr$. Then there exists a profile $\P$ where $f$ returns a choice different from what $\Plr$ returns. There are two possible cases: either $\supp^{xy}_\P > \supp^{yx}_\P$, or $\supp^{yx}_\P > \supp^{xy}_\P$.

\fbox{$\supp^{xy}_\P > \supp^{yx}_\P$} By assumption, $f(\P) = \set{y}$. Now permute the alternatives to obtain a profile $\P'$ where 
$\supp^{xy}_\P = \supp^{yx}_{\P'}$ (and therefore $\supp^{yx}_\P = \supp^{xy}_{\P'}$). Then, by the neutrality of $f$ it follows that $f(\P') = \set{x}$. Notice that $\P$ differs from $\P'$ in having more voters preferring $x$ to $y$. By permuting the agents in $N$ we can then construct a new profile $\P''$ such that $\N^{xy}_{\P''} \subseteq \N^{xy}_{\P}$. That is, we make sure that the same agents ranking $x$ above $y$ in $\P''$ also rank $x$ above $y$ in $\P$.\footnote{Note that this is precisely what is required by the second condition in the definition of monotonicity. Note furthermore that as we are handling only two alternatives, the third condition of the unanimity axiom is trivially satisfied.} By anonymity $f(\P'') = \set{x}$ and by monotonicity $f(\P) = \set{x}$. Contradiction. We thus conclude that  $f(\P) = \set{x}$.

\fbox{$\supp^{yx}_\P > \supp^{xy}_\P$} By assumption, $f(\P) = \set{x}$. An argument identical to the one for the previous case applies to conclude $f(\P) = \set{y}$. We thus obtain a contradiction in this case too.

In both cases we obtain a contradiction. Function $f$ is therefore the plurality rule \eqref{eq:simple}.
\end{proof}
At a high level the proof exploits this argument. Since $m = 2$ and $n$ is odd, $f$ can choose either the minority or the majority alternative. If it chooses the minority alternative, then by anonymity and neutrality the minority alternative should be chosen in all profiles that split $N$ in the same way size-wise (i.e., the cell with ballot $xy$ vs. the cell with ballot $yx$, or the cell with ballot $yx$ vs. the cell with ballot $xy$). However, consistently selecting the minority option goes against monotonicity. Hence $f$ must select the majority alternative. We now look at generalizations of Theorem \ref{th:may}.

\medskip

First of all Theorem \ref{th:may} assumes an odd number of voters. When that is not the case, it is reasonable for the SCF to output both alternatives (a tie). The theorem can be generalized to this setting.
\begin{theorem}[May's theorem with ties \cite{may52set}] \label{th:may2}
Let $\tuple{\N, \A}$ be given such that $m=2$. The plurality rule is the only SCF that is anonymous, neutral and positively responsive.
\end{theorem}
\begin{proof}
See Exercise \ref{ex:may2}.
\end{proof}

\begin{theorem}[May's theorem for quota rules] \label{th:may3}
Let $\tuple{\N, \A}$ be given such that $m=2$. Let $f$ be an SCF that is anonymous, neutral and monotonic. Then there exists $q \in (\frac{n}{2}, n+1]$ such that $f = \Qr_q$.
\end{theorem}
\begin{proof}
See Exercise \ref{ex:may3}.
\end{proof}
Observe that this theorem allows us to obtain Theorem \ref{th:may} as a corollary.

%%%%%%%%%%%%%%%%%%%%%%%%%%%%%%%%%

\subsection{Plurality as maximum likelihood estimator when $m=2$} \label{sec:mle}

Theorem \ref{th:may} above provides a justification of the use of simple majority in contexts involving two alternatives based on the fact that such a rule is the only one satisfying a set of desirable properties. Simple majority can be further justified via a truth-tracking, or epistemic, route. We turn to it in this section.

\subsubsection{Maximum likelihood estimation}

Assume that there exists an objective {\em true} ranking $\pp$ of the two alternatives. Either $\pp = xy$ or $\pp = yx$.\footnote{You can think for instance of $x$ and $y$ as two policies, of which only one is the truly better for the group. Or you can thing of $x$ and $y$ as two candidates for a job, for which only one is the objectively better fit.}  
%That is, the prior probabilities $\Pr(x \succ y) = \Pr(y \succ x) = 0.5$ for $x \neq y \in \A$ are identical. 
Voters are uncertain about which one of the two is the true ranking, but each voter can track the correct option with some given probability $0.5 < p \leq 1$, which is the same for all voters. Such $p$ represents the voters' individual accuracy and corresponds to the conditional probability $\Pr(x \succ_i y \mid \pp = xy)$. That is, the probability that $i$ ranks $x$ above $y$ given that $x$ is (actually) better than $y$. 

So voters observe the state of the world imperfectly and vote accordingly to their imperfect observation. We can then think of each voter's ballot $\p_i$ as a random variable generated by the probability distribution $p$, once the true state of the world $\pp$ has been fixed. So assuming each such vote is independent, a profile is therefore a series of identically distributed independent random variables that can take two values ($xy$ or $yx$).\footnote{You can think of a profile as a sequence of $n$ coin tosses, where each coin corresponds to a voter, and the coin used is a biased coin with bias $p$ towards the correct option.} 
The most likely state of the world is therefore the one that has the highest likelihood of generating the observed profile $\P$, that is: $\argmax_{\pp \in \set{xy,yx}} Pr(\P \mid \pp)$. 

\begin{example}
Assume $\N = \set{1,2,3}$ and $\A = \set{a,b}$. Assume furthermore that $p = \frac{2}{3}$ and that $ab$ and $ba$ have identical priors: $\Pr(\pp = ab) = \Pr(\pp = ba) = 0.5$. Voters can express one of two ballots ($ab$ or $ba$), and assume that the profile we observe is $\P = \tuple{ab, ab, ba}$. 
\begin{align*}
\Pr(\P \mid \pp = ab) & = p^2 \cdot (1-p) 
\end{align*}
while 
\begin{align*}
\Pr(\P \mid \pp = ba) & = p \cdot (1-p)^2 
\end{align*}
So ranking $ab$ is $\frac{p}{1-p} = 2$ times as likely as ranking $ba$ given $\P$.
\end{example}

\begin{remark}
If we assume that the alternatives are equally likely (i.e., that there is an equal prior for $xy$ and $yx$), then the above approach is essentially equal to establishing the probability of each state of the world given the observed profile. Continuing on the previous example, by Bayes rule:
\begin{align*}
\Pr(\pp = ab \mid \P) & = \frac{\Pr(\P \mid \pp = ab) \cdot \Pr(\pp = ab)}{\Pr(\P \mid \pp = ab) \cdot \Pr(\pp = ab) + \Pr(\P \mid \pp = ba) \cdot \Pr(\pp = ba)} \\
& = \frac{p^2 \cdot (1-p) \cdot 0.5}{p^2 \cdot (1-p) \cdot 0.5 + ((1-p)^2 \cdot p) \cdot 0.5} \\
& = \frac{p^2 \cdot (1-p)}{p^2 \cdot (1-p) + (1-p)^2 \cdot p} = \frac{p \cdot (1-p)}{p (1-p) + (1-p)^2} = p.
\end{align*}
This is the probability that $ab$ is the correct state of the world, given the observation of profile $\P$, equals $p$. 
\end{remark}

\subsubsection{A jury theorem}

%Selecting the alternative that maximizes the likelihood of the observed profile
%Being able to always identify, given a profile, the most likely alternative 
%delivers the best possible truth-tracking behavior, under the given error model (i.e., $p$).

Can the maximum likelihood estimation approach described above be implemented through a SCF? Intuitively we would like the SCF to select, given the profile, the alternative corresponding to the most likely correct ranking. It turns out that this is possible, and again such SCF is the plurality rule.

%From a social choice perspective the question is: can we find an SCF $f$ that acts as a maximum likelihood estimator and thereby maximizes the probability of a correct social choice?
%The following theorem shows that the plurality rule is such an SCF.

\begin{theorem}[Condorcet's jury theorem] \label{th:condorcet}
Fix $\A$ such that $m=2$.
% and assume that for all $x \in \A$, $Pr(x) = 0.5$.
Assume furthermore that $0.5 < p \leq 1$ and that, for any $\N$ such that $n$ is odd, each profile for $\tuple{\N, \A}$ is an i.i.d. sequence of random variables $\p_i$ generated by $p$.
Then:
\begin{align}
p_{\Plr}(n) \leq p_{\Plr}(n + 2) & \ \ \mbox{for every odd $n$}  \label{th:growth} \\
p \leq p_{\Plr}(n) & \ \ \mbox{for every odd $n$}  \label{th:non-asymptotic} \\
\lim_{n \to \infty} p_{\Plr}(n)  = 1 & \label{th:asymptotic}
\end{align}
where $p_{\Plr}(n)$ denotes the accuracy of the social choice for $\tuple{N, A}$ determined via plurality.
\end{theorem}
\begin{proof}
First of all observe that a decision taken by plurality is correct if and only if a majority of voters has voted correctly \eqref{eq:simple}. That is, whenever a number $h \geq \frac{n + 1}{2}$ of voters votes correctly. There are 
$\sum_{\frac{n + 1}{2} = h}^{n}  \binom{n}{h}$ such majorities. For each of these the probability that precisely that majority of $h$ voters votes correctly is $p^h \cdot (1 - p)^{n - h}$.
We thus obtain that:
\begin{align}
p_{\Plr}(n) & = \sum_{\frac{n + 1}{2} = h}^{n}  \binom{n}{h} \cdot p^h \cdot (1 - p)^{n - h}
\end{align}
We proceed to prove the three claims of the theorem.

\fbox{\eqref{th:growth}} The claim is a consequence of the following recursive formula:
\begin{align}
p_\Plr(n + 2) & = p_\Plr(n) + \underbrace{(2p - 1) \cdot \binom{n}{\frac{n+1}{2}} \cdot (p \cdot (1-p))^{\frac{n+1}{2}}}_{\phi}.  \label{eq:comb}
\end{align}
Given the assumption $0.5 < p \leq 1$, it is easy to see that $\phi \geq 0$ (and that $\phi > 0$ when in addition $p \neq 1$), from which we can conclude that $p_\Plr(n + 2) > p_\Plr(n)$ as desired. We now proceed to derive \eqref{eq:comb} by a combinatorial argument. We are focusing on profiles of length $n+2$. It is now helpful to think of $p_\Plr(n + 2)$ as the probability that more than half of the first $n+2$ votes in the profile are correct, and of $p_\Plr(n)$ as the probability than more than half of the first $n$ votes are correct. Call the corresponding events $H_{n+2}$ and $H_{n}$. So: $p_\Plr(n + 2) = Pr(H_{n+2})$ and $p_\Plr(n) = Pr(H_{n})$
Consider now the two following events:
\begin{itemize}
\item $H_{n+2} \setminus H_n$ denotes the event that more than half of the first $n+2$ votes, but less than half of the first $n$ votes, are correct;
\item $H_{n+2} \cap H_n$ denotes the event that more than half of the first $n+2$ votes and more than half of the first $n$ votes are correct.
\end{itemize}
Using these events and the standard laws of probability we can rewrite $p_\Plr(n + 2)$ as follows:
\begin{align}
p_\Plr(n + 2) & = Pr(H_{n+2}) \nonumber \\
& = Pr(H_{n+2} \setminus H_n) + Pr(H_{n+2} \cap H_n) \nonumber \\
& = Pr(H_{n+2} \setminus H_n) + Pr(H_{n}) - Pr(H_n \setminus H_{n+2}) \nonumber \\
& = p_\Plr(n) + \underbrace{Pr(H_{n+2} \setminus H_n)}_\alpha - \underbrace{Pr(H_n \setminus H_{n+2})}_\beta \label{eq:precomb}
\end{align}
We proceed to establish $\alpha$ and $\beta$. \fbox{$\alpha$} Event $H_{n+2} \setminus H_n$ occurs if the first $n$ votes contain a narrow incorrect majority (the majority is decided by one voter) and the two subsequent votes are both correct. The first event happens with probability $\binom{n}{\frac{n+1}{2}} \cdot p^{\frac{n-1}{2}} \cdot (1 - p)^{\frac{n+1}{2}}$ and the second with probability $p^2$. We thus obtain:
\begin{align}
\alpha & = \binom{n}{\frac{n+1}{2}} \cdot p^{\frac{n-1}{2}} \cdot (1 - p)^{\frac{n+1}{2}} \cdot p^2  \nonumber \\
& = p \cdot \binom{n}{\frac{n+1}{2}} \cdot (p \cdot (1-p))^\frac{n+1}{2}. \label{eq:alpha}
\end{align} 
\fbox{$\beta$} In the same fashion, event $H_n \setminus H_{n+2})$ occurs if the first $n$ votes contain a correct majority but the first $n+2$ do not, implying the last two votes are incorrect. The first event happens with probability $\binom{n}{\frac{n+1}{2}} \cdot p^{\frac{n+1}{2}} \cdot (1 - p)^{\frac{n-1}{2}}$ and the second with probability $(1-p)^2$. We thus obtain:
\begin{align}
\beta & = \binom{n}{\frac{n+1}{2}} \cdot p^{\frac{n+1}{2}} \cdot (1 - p)^{\frac{n-1}{2}} \cdot (1- p)^2  \nonumber \\
& = (1 - p) \cdot \binom{n}{\frac{n+1}{2}} \cdot (p \cdot (1-p))^\frac{n+1}{2}. \label{eq:beta}
\end{align} 
We can now replace \eqref{eq:alpha} and \eqref{eq:beta} in \eqref{eq:precomb} and obtain \eqref{eq:comb} as desired:
\begin{align*}
p_\Plr(n + 2) & = p_\Plr(n) + p \cdot \binom{n}{\frac{n+1}{2}} \cdot (p \cdot (1-p))^\frac{n+1}{2} - (1 - p) \cdot \binom{n}{\frac{n+1}{2}} \cdot (p \cdot (1-p))^\frac{n+1}{2} \\
& = p_\Plr(n) + (2p - 1) \cdot \binom{n}{\frac{n+1}{2}} \cdot (p \cdot (1-p))^{\frac{n+1}{2}}.
\end{align*}

\fbox{\eqref{th:non-asymptotic}} The claim follows directly from  \eqref{th:growth} once observed that $p = p_{\Plr}(1)$.
%We proceed by induction. For the base case ($n = 1$) we have $p_{\Plr}(1) = p$ and 
%\[
%p_{\Plr}(3) = \sum_{2 = h}^{3}  \binom{3}{h} \cdot p^h \cdot (1 - p)^{n - h} = 3 \cdot p^2 \cdot (1-p) + p^3 = p^2 \cdot (3 - 2p)
%\]
%By the assumption that $0.5 < p \leq 1$, it follows that $p_{\Plr}(1) < p_{\Plr}(3)$ as desired. For the induction step we assume that (IH) $p_{\Plr}(n) < p_{\Plr}(n+2)$ and show that $p_{\Plr}(n+2) < p_{\Plr}(n+4)$.

\fbox{\eqref{th:asymptotic}} By the weak law of large numbers as $n$ grows to infinity the average number of correct votes converges to $p$. As $p > 0.5$ it follows that the probability of a correct majority goes to $1$.
\end{proof}

The theorem is probably the first mathematical analysis of a form of `wisdom of the crowds'. It states that groups become more accurate as they grow larger \eqref{th:growth}, that majorities are more accurate that single individuals \eqref{th:non-asymptotic}, and that infinite groups achieve perfect accuracy \eqref{th:asymptotic}.

%%%%%%%%%%%%%%%%%%

\section{Chapter notes}

The bulk of this chapter is based on the introductions to voting theory provided in \cite[Ch. 2]{comsoc_handbook}, \cite[Ch. 1]{taylor05social} and \cite{endriss11logic}. The Pareto principle was first formulated in \cite{pareto19manuale}.
Theorem \ref{th:may} was first presented in \cite{may52set}, still a beautiful example of the application of the axiomatic method to voting. Theorem \ref{th:condorcet} appeared first in \cite{Condorcet1785} and was rediscovered several times during the last century (e.g., \cite{moore56reliable}). Even though Theorem \ref{th:condorcet} is a well-known result, not many detailed and accessible proofs exists. The proof presented is based on \cite{dietrich20jury}.

%%%%%%%%%%%%%%%%%%%%%%%%%%%%%%%%%%%%%%

\section{Exercises}

\begin{exercise} \label{ex:axioms}
Provide a proof of Fact \ref{fact:axioms}
\end{exercise}

\begin{exercise} \label{ex:may}
Provide a proof of the \LtoR direction of Theorem \ref{th:may}.
\end{exercise}

\begin{exercise} \label{ex:may2}
Provide a proof of Theorem \ref{th:may2}.
\end{exercise}

\begin{exercise} \label{ex:may3}
Provide a proof of Theorem \ref{th:may3}.  \fbox{Hint} You want to show there exists a number $q$ that has two properties: $x$ is the only social choice if at least $q$ voters rank $x$ above $y$; and $q \in(\frac{n}{2}, n+1]$. Now collect in a set $Q$ all numbers $k$ that have the property that if $k$ voters rank $x$ over $y$ then $\set{x}$ is the social choice. $Q$ can be empty, or not. Reason from there.
\end{exercise}

\begin{exercise}[Odd rule]
Consider the following rule.
\begin{VR}\label{vr:odd}
Let $\A = \set{a, b}$. The odd rule is the SCF defined as follows. For every $\P$:
\[
{\tt Odd}(\P) = 
\left\{
\begin{array}{ll}
\set{a} & \mbox{if} \ \supp^{ab}_\P \ \mbox{is odd} \\
\set{b} & \mbox{otherwise}
\end{array}
\right.
\]
\end{VR}
Intuitively, the rule chooses $a$ whenever the size of the set of voters supporting it is odd, otherwise it chooses $b$.
Determine whether the odd rule is: anonymous, neutral, unanimous, positively responsive, non-imposed, resolute. Explain your answers.

\end{exercise}

\begin{exercise}[Iterated voting when $m = 2$]
Consider a series of choice contexts 
\[ 
\tuple{\N_1,\A_1}, \tuple{\N_2, \A_2}, \ldots, \tuple{\N_k,\A_k},
\]
with $k$ large (e.g., $k \geq 2^n$) and such that for every $1 \leq i < k$, $\N_i = \N_{i+1}$, $\A_i = \A_{i+1}$ and such that $m = 2$. For any series $\P_1, \P_2, \ldots, \P_k$ of profiles, a social choice function $f$ would thus determine the series of social choices $f(\P_1), f(\P_2), \ldots, f(\P_k)$. According to your intuition, would plurality be a `fair' way of determining such series of social choices?
%and the series of social choices $\Plr(\P_1), \Plr(\P_2), \ldots, \Plr(\P_k)$ generated by the application of the plurality rule. According to your intuition, is $\Plr(\P_1), \Plr(\P_2), \ldots$ a `fair' set of social choices? 
Justify your answer and if you think the answer is negative devise a voting rule that would behave more `fairly' according to your insights. Again, explain your answer.
\end{exercise}

\begin{exercise}
Consider Theorem \ref{th:condorcet} and determine how its statement should be modified, for all its parts \eqref{th:growth}, \eqref{th:non-asymptotic} and \eqref{th:asymptotic}, once the assumption $0.5 < p \leq 1$ is changed to $p = 0.5$ (first variant), and to $0 \leq p < 0.5$ (second variant). State these two variants of the theorem.
\end{exercise}

%%%%%%%%%%%%%%%%%%%%%%%%%%%%%%%%%%%%%%%%%%%%%%%%%%%%%%%%%%%%

\chapter{Choosing One Out of Many} \label{ch:many-one}

%%%%%%%%%%%%%%%%%%%%%

\section{Beyond plurality}

The classic contributions to social choice theory sparked from the dissatisfaction of how the plurality rule (Rule \ref{vr:plurality}) behaves in social choice contexts when $m > 2$.

\subsection{Plurality selects unpopular options}

\begin{example}[Roman Senate, 1 century b.C.] \label{ex:pliny}
`` \ldots the consul Africanus Dexter had been found slain, and it was uncertain whether it had died at his own hand or at those of his freedmen. When the matter came before the Roman Senate, Pliny wished to acquit them [alternative $a$], another senator moved that they should be banished to an eiland [alternative $b$]; and a third that they should be put to death [alternative $c$]'' \cite[p. 54]{ordeshoek03game}:
\[
\begin{array}{l | lll}
\# 102 & a & b & c \\
\# 101 & b & a & c \\
\# 100 & c & b & a
\end{array}
\]
The standard voting practice in the Roman Senate would have consisted of a staged procedure: first a vote on $a$ (are the freedmen innocent, and therefore to-be-aquitted, or are they guilty?), which would have been lost; second a vote on $b$ vs. $c$, which would have been won by $b$. Pliny, who was chairing the session, favoured $a$, and so demanded a vote by plurality. Plurality would select $a$ as unique social choice which, however, is hugely unpopular: it corresponds to about $\frac{1}{3}$ of the votes of all the senators; and there is one alternative, namely $b$, who is preferred to both $a$ and both $c$ by a strong majority of senators. 
\end{example}

\begin{example}[2002 US Presidential elections, stylized \cite{dasgupta04fairest}] \label{ex:gore}
The example assumes only four kind of voters and abstracts away from the electoral college procedure in use for the election of the US president. It assumes there are 100000 voters (rows indicate multiple of 1000).
\begin{center}
\begin{tabular}{l | llll}
$\# 2$ & Nader  & Gore & Bush & Buchanan \\
$\# 49$ & Gore & Bush & Nader & Buchanan \\
$\# 48$ & Bush & Buchanan & Gore & Nader \\
$\# 1$ & Buchanan & Bush & Gore & Nader
\end{tabular}
\end{center}
Again here plurality would elect Gore, who is receiving a minority of votes. At the same time Bush appears to occupy higher positions than Gore in the voters' rankings.
\end{example}

\subsection{More voting rules}

What seems to be the root of plurality's problems in convincingly aggregating preferences over more than two alternatives is that it disregards all information provided by voters, except for their top choices. Each of the rules discussed in this section puts that extra information at work, in different ways, to establish a `reasonable' social choice.

\subsubsection{Pairwise comparison rules}

One piece of information that appears to be relevant in the above examples is how often an alternative $x$ `beats' another alternative $y$ in a pairwise contest, that is, whether there exists a larger support for $x$ over $y$ than there is for $y$ over $x$. 
We can easily tabulate such information as a {\em pairwise comparison} $m \times m$ matrix $[ \supp^{xy}_\P]$ with $x, y \in \A$. 
Clearly, for $x \neq y$, $\supp^{xy}_\P = n- \supp^{yx}_\P$ and, for any $x \in \A$, $\supp^{xx} = 0$.\footnote{This method of representing relevant election information was known already to R. Lull \cite{llull75artifitium} in the 13th century.}
\begin{example}
The pairwise comparisons matrix for the election described in Example \ref{ex:gore} is: 
\[
 \begin{matrix}
		&  \mbox{Buchanan} & \mbox{Bush} & \mbox{Gore} & \mbox{Nader} \\
\mbox{Buchanan} &         0        &         1               &           49            &        49                 \\
\mbox{Bush}   &             99        &            0               &         49                &        98               \\
\mbox{Gore}   &          51           &           51              &            0             &      98                     \\
\mbox{Nader} &        51             &             2              &         2                &             0              \\
\end{matrix}
\]
\end{example}

Based on such information we can then identify which alternatives are considered better than which other alternatives in pairwise comparisons and thereby possibly identifying the `best' alternatives. 

\begin{definition}[Net preference]
Let $\P$ be a profile for context $\tuple{\N, \A}$. For any pair of distinct alternatives $xy$, the {\em net preference} for $x$ over $y$ in $\P$ is
\begin{align*}
\net_\P(xy) & = \supp^{xy}_\P - \supp^{yx}_\P.
\end{align*}
\end{definition}
Note that $\net_\P(xy) = \supp^{xy}_\P - \supp^{yx}_\P = 2 \cdot \supp^{xy}_\P-n$.

\begin{definition}[Pairwise majority tournament] \label{def:tournament}
The pairwise majority tournament (or {\em majority graph}) of $\P$ is the graph $\tuple{\A, \succ_\P^\net}$ where $\pp^\net_\P \subset A^2$ is a binary relation such that $x \succ^\net_\P y$ iff $\net_\P(xy) > 0$. The weighted pairwise majority tournament (or {\em weighted majority graph}) is the graph $\tuple{\A, \succ_\P^\net}$ where each edge $xy \in \succ^\net_\P$ is labeled by $\net_\P(xy)$.
When the profile is clear from the context we will omit reference to it and write simply $\net(xy)$ and $\succ^\net$. The {\em weak majority tournament}, allowing for ties, can be naturally defined as: $x \succeq^\net_\P y$ iff $\net_\P(xy) \geq \net_\P(yx)$.
\end{definition}

Given a majority tournament, there is therefore one natural way to identify a possible social choice: check whether there exists an alternative that beats any other alternative in the tournament.

\begin{definition}[Condorcet winner]
Let $\tuple{\A, \succ^\net}$ be the majority graph of $\P$. An alternative $x \in \A$ is the {\em Condorcet winner} of $\P$ if for all $y \in \A \setminus \set{x}$, $x \succ^\net_\P y$. It is a {\em weak Condorcet winner} of $\P$  if for all $y \in \A \setminus \set{x}$, $x \succeq^\net_\P y$.
\end{definition}
So a Condorcet winner is an alternative that beats every other alternative in a pairwise comparison---in graph-theoretical terms it is a `source' in the majority graph. As such, it can be considered the natural social choice, at least if we care about majorities.
Notice that relation $\succ^\net_\P $ is irreflexive and complete when $n$ is odd, for any profile $\P$. However, it is not transitive in general. It follows that a Condorcet winner may not exist.

\begin{example} \label{ex:tour}
The majority graph of the profile in Example \ref{ex:pliny} is: $b \succ^\net a \succ^\net c$. Alternative $b$ is therefore the Condorcet winner of the profile. The majority graph of the profile in Example \ref{ex:gore} is:
\[
\mbox{Gore} \succ^\net \mbox{Bush} \succ^\net \mbox{Nader} \succ^\net \mbox{Buchanan}
\]
\end{example}

\begin{example}[Condorcet paradox \cite{Condorcet1785}] \label{ex:condorcet}
Three instantiations of the paradox for $\tuple{\set{1,2,3}, \set{a,b,c}}$ (left), $\tuple{\set{1, \ldots, 303}, \set{a, b, c}}$ (center), $\tuple{\set{1,2,3}, \set{a,b,c, d}}$ (right):
\begin{align*}
\P_1 = 
\begin{array}{l | lll}
1 & a & b & c \\
2 & b & c & a \\
3 & c & a & b
\end{array}
& &
\P_2 = 
\begin{array}{l | lll}
102 \times & a & b & c \\
101 \times & b & c & a \\
100 \times & c & a & b
\end{array}
& &
\P_3 = 
\begin{array}{l | llll}
1 & a & b & c & d \\
2 & b & c & a & d\\
3 & c & a & b & d
\end{array}
\end{align*}
The pairwise majority tournament all the above profiles contains a $3$-cycle $a \succ^\net b \succ^\net c \succ^\net a $. These are also called {\em majority cycles} and can arise whenever $m > 2$.
\end{example}
Finally it is worth knowing that any tournament (i.e., irreflexive, complete binary relation) can be generated via Definition \ref{def:tournament} by some profile, under the assumption that $n$ is odd (McGarvey's theorem \cite{mcgarvey53theorem}).

\medskip

We introduce now two rules based on (unweighted) majority graphs.

\begin{VR}[Condorcet] \label{vr:condorcet}
The Condorcet rule is the SCF defined as follows. For every profile $\P \in \LO(\A)^n$:
\begin{align*}
\Con(\P) & = 
\left\{
\begin{array}{ll}
\set{x} & \mbox{if} \ \forall y \in \A \setminus \set{x}, x \succ^\net_\P y \\
\A & \mbox{otherwise} 
\end{array}
\right.
\end{align*}
\end{VR}
So the Condorcet rule selects a Condorcet winner when it exists and let otherwise all alternatives tie.

\begin{VR}[Copeland \cite{llull75artifitium,copeland75reasonable}] \label{vr:copeland}
The Copeland rule is the SCF defined as follows. For any profile $\P \in \LO(\A)^n$:
\begin{align*}
\Cop(\P) & = \argmax_{x \in \A} C_\P(x)
\end{align*}
where 
\[
C_\P(x) =  \size{\set{y \in \A \mid x \succ^\net_\P y }} - \size{\set{y \in \A \mid y \succ^\net_\P x }}. 
\]
This value is called the {\em Copeland score} of $x$ (in $\P$).
\end{VR}
So the Copeland score consists of the number of times an alternative beats other alternatives, minus the number of times it is beaten by other alternatives (in the majority graph of the corresponding profile). The social choice is then the set of alternatives with the highest Copeland score. Notice that the rule disregards the margins whereby an alternative beats, or is beaten by, other alternatives: the Copeland score can be computed by just looking at the majority graph of a profile, disregarding the net preference information. In other words, it disregards the relative positions of alternatives in the voters' rankings. This type of information becomes available in the next class of rules we are going to discuss.

\begin{example}
Consider again Example \ref{ex:condorcet}. We have that: $\Con(\P_1) = \Cop(\P_2) = \set{a,b,c}$ but $\Con(\P_3) = \set{a, b, c, d}$ while $\Cop(\P_3) = \set{a, b, c}$.
\end{example}

\begin{remark}
Another voting rule based on (unweighted) majority graphs is the Slater rule \cite{slater61inconsistencies}: given the majority graph $\tuple{A, \succ^\net}$ find the strict linear order $\succ$ that minimizes the number of edges one needs to switch in  $\succ^\net$ to get $\succ$. The social choice is the singleton consisting of the top alternative in $\succ$.
\end{remark}

\begin{remark}
There are voting rules that need the information of the weighted majority graph to be computed. An example is the Ranked-pairs rule: given the weighted majority graph $\tuple{A, \succ^\net, \net}$ order the edges of relation $\succ^\net$ by the size of their weight (largest to smallest), then construct a linear order over $A$ by starting from the edge with the largest weight and proceeding to add edges following their weights unless a cycle is formed. The social choice is the top alternative in the constructed order. 
\end{remark}

\subsubsection{Positional scoring rules}

Positional scoring rules are a class of SCFs that determine the social choice by assigning points to alternatives depending on how often the alternative occupies a given position in the voters' ballots. The social choice is the alternative collecting the most points.

\begin{definition}[Positional scoring rules] \label{def:scoring}
A positional scoring rule $f$ is an SCF, for $\tuple{\N, \A}$, such that there exists a vector of weights ({\em score vector}) $\w = \tuple{w_1, w_2, \ldots, w_m}$ and, for all $\P \in \LO(\A)^n$ 
\[
f(\P) = \argmax_{x \in \A} \sum_{i \in \N} w_{i(x)}
\]
where $i(x)$ denotes $x$'s position in the linear order $\p_i$. 
\end{definition}
Observe that plurality (Rule \ref{vr:plurality}) can also be considered a scoring rule, with score vector $\tuple{1, {\underbrace{0, \ldots, 0}_{m \ \mbox{times}}}}$. 

The most famous, and most widely deployed, positional scoring rule is the Borda rule.\footnote{The method carries the name of Jean-Charles de Borda, French mathematician and engineer \cite{borda81memoires}. The method already appeared in the writings of Nicholas of Cusa in the 15th century.}
\begin{VR}[Borda] \label{vr:borda}
The Borda rule $\Br$ is the positional scoring rule with score vector 
\[
\w = \tuple{m-1, m-2, \ldots, 0}.
\] 
For every alternative $x$ and profile $\P$, $\sum_{i \in \N} w_{i(x)}$ is called the {\em Borda score} of $x$ (in $\P$).
\end{VR}
In the Borda method an alternative $x$ therefore collects, for each voter, as many points as the alternatives that that voters places under $x$ in her ballot. Such points are then summed up across voters to obtain the Borda score of the alternative in the given profile.

\begin{example}[Borda vs. Condorcet]
Consider again the profile in Example \ref{ex:gore}. We have seen that the Condorcet winner in that profile is Gore (Example \ref{ex:tour}). However, the Borda rule selects Bush:
\begin{align*}
\sum_{i \in \N} w_{i(\mbox{Bush})} = 3 \cdot 48 + 2 \cdot 50 + 1 \cdot 2 & & \sum_{i \in \N} w_{i(\mbox{Gore})} = 3 \cdot 49 + 2 \cdot 2 + 49 \\
\sum_{i \in \N} w_{i(\mbox{Nader})} = 3 \cdot 2 + 1 \cdot 49 & & \sum_{i \in \N} w_{i(\mbox{Buchanan})} =  2 \cdot 48
\end{align*}
where the relevant score vector is $\tuple{3, 2, 1, 0}$.
\end{example}
So the Borda rule may fail to select a Condorcet winner when it exists. This can be shown to generalize to all positional scoring rules.

\medskip

Other examples of positional scoring rules follow:

\begin{VR}[Veto] \label{vr:veto}
The veto rule $\Vr$ is the positional scoring rule with score vector $\w = \tuple{\underbrace{1, \ldots, 1}_{m-1 \ \mbox{times}}, 0}$. 
\end{VR}

\begin{VR}[$k$-Approval] \label{vr:approval}
The $k$-Approval rule $\Ar_k$, with $1 \leq k \leq m$, is the positional scoring rule with score vector $\w = \tuple{\underbrace{1, \ldots, 1}_{k \ \mbox{times}}, \underbrace{0, \ldots, 0}_{m - k \ \mbox{times}}}$. 
\end{VR}

So plurality is the $1$-Approval rule, and veto is the $m-1$-Approval rule.

An area where positional scoring rules are widely applied are sport or music competitions, such as the F1 championship or the Eurovision Song Contest competition. For example, the winner of the F1 championship is determined via a positional scoring rule based on the ranking of each pilot in a fixed number of races. Notice that each race here acts as a voter: it determines a linear order over the pilots. Over the years various score vectors have been used for the computation of winners, for example: in the 1961-1990 seasons the vector $\tuple{9,6,4,3,2,1, 0, \ldots, 0}$ was used; in the 1991-2002 seasons, the vector $\tuple{10,6,4,3,2,1, 0, \ldots, 0}$; and in the 2003-2208 seasons, the vector $\tuple{10, 8, 6, 5 , 4, 3 , 2, 1, 0, \ldots, 0}$. In song competitions, like the Eurovision Song Contest, juries (and the public) act as voters. The score vector of the Eurovision Song Context is currently $\tuple{12,10,8, 7, 6, 5, 4,3,2,1, 0, \ldots, 0}$.

\subsubsection{Multiround rules}

This third class of rules is based on the idea that ballots can be revised in an iterated manner by removing, at each round in the process, the `least popular' alternatives from the ballots, until one alternative achieves majority.

Before introducing these rules we need some auxiliary notation. Let $\P$ be a profile and $X \subseteq A$ a set of alternatives. Then the profile restricted to $X$ is $\P|_{X} = \tuple{\p_1|_X, \ldots,  \p_n|_X}$, where each $\p_i|_X$ is the restriction of $\p_i$ to the alternatives in $X$. Then let $\sigma_\P: 2^{\A} \to 2^{\A}$ be a set-transformer defined as follows:
\[
\sigma_\P(X) = X \setminus \set{x \in \A \mid x \ \mbox{has lowest plurality score in} \ \P|_X}.
\]
That is, given a set of alternatives $X$, $\sigma_\P(X)$ removes from $X$ the alternatives that are ranked highest by the smallest number of voters. For $k \in \mathbb{N}$ we denote by $\sigma^{k}_\P$ the $k^{\mathit{th}}$-fold iteration of $\sigma_\P$.

\begin{VR}[Plurality with run-off] \label{vr:ploff}
The plurality with run-off rule is the SCF defined as follows, for any profile $\P$:
\begin{align*}
\PlrO(\P) & =
\left\{
\begin{array}{ll}
\set{x} & \mbox{if} \ \size{\set{i \in \N \mid \max_{\p_i}(\A) = x }} \geq \frac{n+1}{2}  \\
\Plr(\P|_\mathit{Top2}) & \mbox{otherwise}
\end{array}
\right.
\end{align*}
where $\mathit{Top2} \subseteq A$ is the set of the $2$ alternatives in $\A$ with the two highest plurality scores in $\P$.
\end{VR}

\begin{example}[Non-monotonicity]
Consider the two profiles
\begin{align*}
\P = 
\begin{array}{l | lll}
\# 8 & a & b & c \\
\# 10 & c & a & b \\
\# 7 & b & c & a
\end{array} 
& & 
\P' = 
\begin{array}{l | lll}
\# 6 & a & b & c \\
\# 2 & c & a & b \\
\# 10 & c & a & b \\
\# 7 & b & c & a
\end{array}
\end{align*}
We have that $\PlrO(\P) = \set{c}$ ($b$ is eliminated in the first round, and $c$ wins in the second round against $a$ by a net preference of $17-8$) but $\PlrO(\P') = \set{b}$ ($a$ is eliminated in the first round, and $b$ wins in the second round against $b$ by a net preference of $13-12$) even though in $\P'$ $2$ voters have ranked $c$ higher than in $\P$.

\end{example}

\begin{VR}[Single transferable vote] \label{vr:STV}
The single transferable vote rule (STV) is the SCF defined as follows, for any profile $\P$:
\begin{align*}
\STV(\P) & = \sigma^{k^\star - 1}_\P(A)
%\gfp.\sigma_\P
\end{align*}
where $k^*$ is the step at which all alternatives are removed: $\sigma^{k^\star}_\P(A) = \emptyset$. 
%\footnote{The fixpoint computation consists of this stream of sets: $\A, \sigma_\P(A), \sigma_\P(\sigma_\P(A)), \ldots$. As $\A$ is finite a fixpoint must exist. More generally, it is easy to see that $\sigma_\P$ is a monotonic set-%transformer, hence the Knaster-Tarski fixpoint theorem applies (cf. \cite{davey90introduction}).}
\end{VR}
Intuitively, STV iteratively removes the alternatives with lowest plurality score until all alternatives are removed. So STV is still based on the plurality score, but it extracts as much information as possible from the ballot profile, thereby reducing the amount of votes `wasted'. By means of comparison, in a standard plurality election normally more than half the votes would be `wasted' in the sense of supporting losers of the election. STV obviates to this problem by iteratively removing plurality losers and recomputing the plurality score based on the preferential ballots. So the alternatives that are removed last constitute the social choice, which therefore consists of either one alternative with a majority support, or by alternatives with equal plurality score.\footnote{Several variants to this definition exist, mostly focusing on what to do when ties exist among the alternatives with lowest scores.}
The STV rule goes also under other names, e.g., Hare rule \cite{hare59treatise},\footnote{In 1862 John Stuart Mill referred to it as ``among the greates improvements yet made in the theory and practice of government''. It is popular among electoral reform groups and it is widely deployed to elect representatives (among others in Ireland, Northern Ireland, Australia).} instant run-off voting, ranked choice voting.

%%%%%%%%%%%%%%%%%%%%%%%%

\subsection{More axioms for SCFs}

\begin{definition} \label{def:more}
Let $\tuple{\N, \A}$ be given. An SCF $f$ is:
\begin{description}

\item[Non-dictatorial] iff there exists no $i \in \N$ s.t. for all $\P \in \LO(\A)^n$, $f(\P) = \Dr_i(\P)$ (cf. Rule \ref{vr:dictatorship}). Otherwise $f$ is said to be dictatorial with $i$ being {\em the dictator}.

{\em Intuitively}, there exists no voter that can unilaterally determine the social choice. 

\item[Condorcet-consistency] iff for all $\P \in \LO(\A)^n$, if $x$ is the Condorcet winner of $\P$, then $f(\P) = \set{x}$.

{\em Intuitively}, the rule agrees with the Condorcet rule on all profiles that give rise to transitive majority graphs.

\item[Independent] iff  for any two profiles $\P, \P' \in \LO(\A)^n$ and two alternatives $x, y \in \A$, if $\N^{xy}_\P = \N^{xy}_{\P'}$ and $x \in f(\P)$ but $y \not\in f(\P)$ then $y \not\in f(\P')$.

{\em Intuitively}, whether one alternative belongs to the social choice while the other does not depends only on how voters rank the two alternatives, and not on how they rank other alternatives.
%if in two profiles the same set of voters rank two alternatives in the same way, and only one of these alternatives belongs to the social choice of the first profile, the other alternative cannot belong to the social choice of the second %profile. 

\item[Liberal] iff for all $i \in \N$ there exists $x \neq y \in \A$ s.t. for every profile $\P \in \LO(\A)$, $i \in \N^{xy}_\P$ implies $y \not\in f(\P)$ and $i \in \N^{yx}_\P$ implies $x \not\in f(\P)$.  Each such agent $i$ is said to be {\em two-way decisive} on $x$ and $y$.

{\em Intuitively}, each agent should be able to unilaterally veto, for at least two alternatives, whether the alternative is part of the social choice.

\end{description}
\end{definition}

Axioms can help understand the different behavior of different voting rules by looking at what rules satisfy what axioms. See Exercise \ref{ex:rules_axioms}.

\subsection{Impossibility results for SCFs: examples}

As Table \ref{tab:rules_axioms} makes evident, the rules we have considered so far satisfy some of the desirable properties---the axioms---we came up with, but none satisfies all. The question then arises of whether such `ideal' social choice functions could actually be found. The upshot of the next section is that this is not the case, and the key means to show that is through so-called {\em impossibility} results: theorems stating the inexistence of functions with a given set of properties. Before moving to the next section and discuss one of the main impossibility results, we illustrate here `style' of such results with two simple(r) examples.

\begin{fact} \label{fact:warmup}
Let $\tuple{\N, \A}$ be given such that $n = m = 2$. There exists no SCF that is anonymous, neutral and resolute.
\end{fact}

\begin{theorem}
%[\cite{sen70impossibility}] 
\label{th:sen}
Let $\tuple{\N, \A}$ be given such that $m > 2$ and $n \geq 2$. There exists no SCF that is Pareto (recall Definition \ref{def:resp_axioms}) and liberal.
\end{theorem}
\begin{proof}
See Exercise \ref{ex:sen}.
\end{proof}

%%%%%%%%%%%%%%%%%%%%%%%%%%%%%%%%%%%%%%%%%%%%

\section{There is no obvious social choice function when $m > 2$}

\subsection{Social preference functions}

The rules discussed in the previous section not only identify a social choice for each profile, but actually determine a weak `social preference' over the set of alternatives. That is, they can be viewed as {\em social preference functions}\footnote{A more common terminology is `social welfare function'.} mapping profiles of individual linear orders to total preorders (that is, linear orders admitting ties):
\begin{align}
F: \LO(\A)^n \to \TP(\A). \label{eq:SPF} 
\end{align}
where $\TP(\A)$ denotes the set of all total preorders\footnote{That is, reflexive, transitive and total binary relations.} on the set of alternatives. So, for any profile $\P$, we refer to $F(\P) = \p^F_\P$ as the social preference determined by $F$ in profile $\P$, and to $\pp^F_\P$ as its asymmetric (strict) part.

\smallskip

It is worth stressing that SCFs and SPFs are closely related:
\begin{remark}[Correspondence between SCFs and SPFs]
Let $\tuple{\N, \A}$ be given. 
\begin{enumerate}[a)]
\item For every SPF $F$,  $\max_{F( \cdot )}(\A)$ is the SCF that selects, for any profile $\P \in \LO(\A)^n$, the top alternatives in  $\p^F_\P$.
\item For every SCF $f$, the linear order 
\[
f(\P) \succ f(\P|_{\A \setminus f(\A)}) \succ f(\P_{f(\P|_{\A \setminus f(\A)})}) \succ \ldots.
\]
defines a social preference function $\p^F_\P$ where the tied top-ranked alternatives are the alternatives in $f(\P)$, the tied second-ranked alternatives are the alternatives in $f(\P|_{\A \setminus f(\A)})$ and so on. Observe that the length of $\p^F_\P$ is clearly at most $m-1$ when $f$ is resolute.
\end{enumerate}
\end{remark}
Intuitively, the social choice corresponds to the best alternatives in such a social preference. Vice versa, the social preference can be induced by iterated social choices on sets of alternatives from which the winners of the previous social choice have been removed. 

\medskip

It is worth observing right away that May's and Condorcet's theorems (Theorems \ref{th:may} and \ref{th:condorcet}) for the setting $m = 2$ can be straigthorwardly formulated for the setting of SPFs.
In general, the axioms we studied so far can all be adapted to the SPF setting. Here we present three, that will be used for the impossibility result discussed later in this section.

\begin{definition}
Let $\tuple{\N, \A}$ be given. An SPF $F$ is:
\begin{description}

\item[Non-dictatorial] iff there exists no $i \in \N$ s.t. $\forall \P \in \LO(\A)^n$, $F(\P) =  \p_i$. Otherwise $f$ is said to be dictatorial with $i$ being {\em the dictator}.

{\em Intuitively}, there exists no voter whose ballot is always identical to the social preference. 

\item[Pareto] iff for any profile $\P \in \LO(\A)^n$ whenever $\N^{xy}_\P = \N$, $x \pp^F_\P y$.
%$xy \in F(\P)$

{\em Intuitively}, if everybody agrees on the relative ranking of two alternatives that ranking should be part of the social preference.

\item[Independent of irrelevant alternatives (IIA)] iff  for any two profiles $\P, \P' \in \LO(\A)$ and two alternatives $x, y \in \A$, if $\N^{xy}_\P = \N^{xy}_{\P'}$ then $x \p^F_\P y$ iff $x \p^F_{\P'} y$
%$xy \in F(\P)$ iff $yx \in F(\P')$.

{\em Intuitively}, if two profiles are identical with respect to the relative ranking of two alternatives, then the relative ranking of the two alternatives in the the social preferences for those two profiles is the same.

\end{description}
\end{definition}

The following proposition shows how properties such as the above ones may interact in ways that may be unexpected at first. The proposition will also be of importance in the discussion of Arrow's theorem in the next section.
\begin{proposition} \label{prop:linear}
Let $F$ be a SPF for a $\tuple{\N, \A}$ s.t. $m > 2$. If $F$ satisfies Pareto and IIA, then $F$ always outputs a linear order.
\end{proposition}
\begin{proof}
Assume towards a contradiction that there exists a profile $\P$ in which $x$ and $y$ are tied in $F(\P)$---$x \sim^F_\P y$. Now consider a new profile $\P'$ such that $\N^{xy}_\P = \N^{xz}_{\P'} \cap \N^{zy}_{\P'}$ and $\N^{yx}_\P = \N^{zy}_{\P'} \cap \N^{zx}_{\P'}$. That is, in $\P'$ the voters that were ranking $x$ over $y$ now place an alternative $z$ between $x$ and $y$, while the voters that were ranking $y$ over $x$ now place an alternative $z$ above both $y$ and $x$. By IIA we get $x$ and $y$ tied in $F(\P')$---$x \sim^F_{\P'} y$---and by Pareto we get $z \pp^F_{\P'} y$ and therefore $z \pp^F_{\P'} x$. 

Now consider a different profile $\P''$ where we choose a different placement for $z$: $\N^{xy}_\P = \N^{xz}_{\P''} \cap \N^{yz}_{\P''}$ and $\N^{yx}_\P = \N^{yz}_{\P''} \cap \N^{zx}_{\P''}$. That is, in $\P''$ the voters that were ranking $x$ over $y$ now place the alternative $z$ below both $x$ and $y$, while the voters that were ranking $y$ over $x$ now place an alternative $z$ between the two. By IIA we get again $x \sim^F_{\P''} y$ and by Pareto we get $y \pp^F_{\P''} z$ and therefore $x \pp^F_{\P''} z$. However $\N^{xz}_{\P'} = \N^{xz}_{\P''}$, but $z \pp^F_{\P'} x$ and $x \pp^F_{\P''} z$, which violates IIA. Contradiction.
 \end{proof}

The lemma tells us that SPFs that are Pareto and IIA, when there are at least three alternatives, map profiles of linear orders to linear orders, that is, they do not allow for ties in the social preference.

%%%%%%%%%%%%%%%%%%%%%%%%%%%%%

\subsection{Arrow's theorem}

Arrow's theorem establishes that it is impossible to aggregate the preferences of a finite set of voters while respecting Pareto, independence of irrelevant alternatives and non-dictatorship.  The three axioms are inconsistent.\footnote{The enormous impact that this theorem had in social choice, economics and political theory in general was well summarized by Paul Samuelson (like Arrow, a winner of the Nobel Memorial Prize in Economic Science):``Arrow's devastating discovery is to mathematical politics what Kurt G\"odel's 1931 impossibility-of-proving-consistency theorem is to mathematical logic''.}
\begin{theorem}[Arrow's theorem \cite{arrow50difficulty,arrow63social}] \label{th:arrow}
Let $F$ be a SPF for a $\tuple{\N, \A}$ s.t. $m > 2$. $F$ satisfies Pareto and IIA if and only if it is a dictatorship.
\end{theorem}
Another way to look at the theorem is that of a characterization of the dictatorship rule, exactly like May's theorem is a characterization of the plurality rule in the $m = 2$ context: dictatorships are the only SPFs that are Pareto and IIA. In a way the theorem shows that there is no `obvious' method---in the sense of being independent of irrelevant alternatives and Pareto---to aggregate preferences on more than $2$ alternatives.

\subsubsection{Proof strategy and decisive coalitions}

The proof of the theorem presented here relies on one main definition and three lemmas. The central definition in the proof of Arrow's theorem is that of a decisive coalition of voters.
\begin{definition}[Decisive coalition] \label{def:decisive}
Let $F$ be a SPF (for a given $\tuple{\N, \A}$), and $x, y \in \A$. A coalition $C \subseteq \N$ is {\em decisive for} $x$ {\em over} $y$ (or $xy$-decisive), under $F$, if
\[
\forall \P \in \LO(\A)^n:  \ \mbox{if} \ C \subseteq \N^{xy}_\P \ \mbox{then} \ x \pp^F_\P y.
%xy \in F(\P).
\] 
A coalition $C \subseteq \N$ is decisive if it is decisive for every pair of alternatives.
The set of $xy$-decisive coalitions (under $F$) is denoted $\Win^{xy}_F$ (or simply $\Win^{xy}$ when $F$ is clear from the context). The set of decisive coalitions (under $F$) is denoted $\Win$.
\end{definition}
Intuitively, a set of voters is decisive for $x$ over $y$ whenever, if they all agree with ranking $x$ over $y$, then $x$ is ranked over $y$ in the social preference. 

\medskip

The proof then relies on three main lemmas. 
The first lemma shows that if a coalition is decisive for a pair of alternatives, then it is decisive for all alternatives.

The second one shows that if SPFs are Pareto and IIA, then their set of decisive coalitions $\Win$ takes the form of a so-called ultrafilter. Ultrafilters are collections of sets and were originally introduced in \cite{cartan37filtres} to capture a handful of properties characterizing the notion of  `large set', like: i) the largest set is a large set; ii) a set is large iff its complement is not large; iii) the intersection of two large sets is large. 

The third lemma is a fact about ultrafilters constructed on finite sets (like the set of voters $\N$), and it states that every finite ultrafilter contains a singleton. It follows that if $\Win$ is an ultrafilter, it must be finite since $\N$ is finite, and it must therefore contain a singleton decisive coalition, that is, a dictator. The intuition behind the use of ultrafilters in voting theory is that a notion of `large set' can naturally be used to define SFPs responding to the rough intuition: issue $x$ is ranked above $y$ in the social preference if and only if there is a `large set' of individuals---a `large coalition'---supporting it.

\subsubsection{The three lemmas}

\begin{lemma}[Contagion lemma] \label{lemma:contagion}
Let $F$ be a SPF (for a given $\tuple{\N, \A}$) which is Pareto and IIA. If $C \in \Win^{xy}$ for some $x,y \in A$, then $C \in \Win$.
\end{lemma}
\begin{proof}
We want to show that if $C \in \Win^{xy}$ for some $x,y \in A$ then $C \in \Win^{wz}$ for all $w,z \in A$. So consider any profile $\P$ where each voter $i$ in $C$ holds preferences satisfying $w \succ_i x \succ_i y \succ_i z$ and any other voters $j$ holds preferences satisfying $w \succ_j x$ and $y \succ_j z$:
\begin{align*}
\P  =
\begin{array}{c | lllllllll}
C & \ldots & w & \ldots & x &  \ldots & y & \ldots & z & \ldots \\
\overline{C}    & & \ldots & w  & \ldots &  x  & \ldots & \\
		      &   & \ldots & y  & \ldots &  z  & \ldots & 
\end{array} 
\end{align*}
Notice that this specification of $\P$ leaves it open to how voters in $\overline{C}$ rank $w$ w.r.t. $z$ and $x$ w.r.t $y$.
Now $C \subseteq \N_{\P}^{wz}$ and, since $C$ is $xy$ decisive, we therefore have that $x \succ^F_\P y$ and, by Pareto, we have that $w \succ^F_\P x$ and $y \succ^F_\P z$. By transitivity we thus have that $w \succ^F_\P z$. Now consider any profile $\P'$ such that $C \subseteq \N_{\P'}^{wz}$. By the specification we gave for $\P$, there exists a profile $\P$ such that $\N_{\P}^{wz} = \N_{\P'}^{wz}$. By IIA, $w \succ^F_{\P'} z$. We can then conclude that $C \in \Win^{wz}$, as desired.
\end{proof}

\begin{lemma}[Ultrafilter lemma] \label{lemma:ultrafilter}
Let $F$ be a SPF (for a given $\tuple{N,A}$), that satisfies Pareto and IIA.
The set $\Win$ of decisive coalitions (for $F$) is an {\em ultrafilter} over $\N$, that is:
\begin{enumerate}[i)]
%\item $\tuple{\Win, \subseteq}$ is a poset;

\item $\N \in \Win$, i.e., the set of all individuals is a decisive coalition;

\item $C \in \Win$ iff $\overline{C} \not\in \Win$, i.e., a coalition is decisive if and only if its complement is not;

%\item if $C \in \Win$ and $C \subseteq C'$ then $C' \in \Win$, i.e., if a coalition is decisive, all coalitions containing it are also decisive;

\item $\Win$ is closed under finite intersections: if $C, D \in \Win$ then $C \cap D \in \Win$, i.e., if two coalitions are winning then the individuals they have in common form a winning coalition.
\end{enumerate}
\end{lemma}

\begin{proof}
\begin{enumerate}[i)]

\item The claim is a direct consequence of the assumption that $F$ satisfies Pareto.

\item \LtoR Suppose, towards a contradiction, that $C, \overline{C} \in \Win$. Consider now a profile $\P$ where $C = \N_\P^{xy}$ and $\overline{C} = \N_\P^{yx}$. This profile must exist as SPFs admit any profile in $\LO(\A)^n$ as input, and be
such that $xy, yx \in F(\P)$. But $F(\P)$ must be a linear order by Proposition \ref{prop:linear}. Contradiction.

\RtoL Assume $\overline{C} \not\in \Win$. Then there exists a pair of (distinct) alternatives $yx$ such that $\overline{C} \not\in \Win^{yx}$. So assume, for $x \neq y \in \A$, that $\overline{C} \not\in \Win^{yx}$. Then, by Definition \ref{def:decisive}, there exists a profile $\P$ such that $\overline{C} \subseteq \N^{yx}_\P$ and $x \succ^F_\P y$. Such profile can be depicted as follows:
\begin{align*}
\P  =
\begin{array}{c | llllll}
C' & \ldots & x &  \ldots & y & \ldots \\
C''    & \ldots & y &  \ldots & x & \ldots \\
\overline{C} & \ldots & y &  \ldots & x & \ldots \\
\end{array} 
\end{align*}
where $C = C' \cup C''$ and $C''$ may be empty. Now consider the following variant of $\P$ constructed by placing a third alternative $x$ in different positions in the individual ballots:
\begin{align*}
\P' & =
\begin{array}{c | llllllll}
C' & \ldots & x &  \ldots & y & \ldots & z & \ldots \\
C''   & \ldots & y &  \ldots & z & \ldots & x & \ldots \\
\overline{C} & \ldots & y &  \ldots & z & \ldots & x & \ldots  \\
\end{array} 
\end{align*}
The social preference $F(\P')$ for this profile is such that: $x \succ^F_{\P'} y$ by IIA; $y \succ^F_{\P'} z$ by Pareto; $x \succ^F_{\P'} z$ by the definition of SPF \eqref{eq:SPF}. By IIA, in any profile $\P''$ such that $\N^{xz}_{\P'} = \N^{xz}_{\P''}$ we have that $x \succ^F_{\P''} z$. It follows that $C' \in \Win^{xz}$. Since, by construction, $C' \subseteq C$ we also obtain $C \in \Win^{xz}$. Finally, by Lemma \ref{lemma:contagion}, we conclude $C \in \Win$ as desired.

%\item The claim follows directly from the definition of decisive coalition (Definition \ref{def:decisive}).

\item Assume towards a contradiction that $C, D \in \Win$ and $C \cap D \not\in \Win$. By the previous item, $\overline{C \cap D} \in \Win$. Construct a profile $\P$ with the following features:\footnote{Note the resemblance with the Condorcet paradox (Example \ref{ex:condorcet}).}
\[
\begin{array}{l | lllllll}
C \cap D & \ldots & x & \ldots & y & \ldots & z & \ldots \\
D \setminus C & \ldots & y & \ldots & z & \ldots & x & \ldots \\
C \setminus D & \ldots & z & \ldots & x & \ldots & y & \ldots \\
\overline{C \cup D} & \ldots & z & \ldots & y & \ldots & x & \ldots \\
\end{array} 
\]
We have that:
\begin{itemize}
\item $(C \cap D) \cup (C \setminus D) = C$, which is decisive by assumption. So, as for all $i \in C$ $x \pp_i y$, it follows that $x \succ^F_\P y$;
\item $(C \cap D) \cup (D \setminus C) = D$, which is decisive by assumption. So, as for all $i \in D$ $y \pp_i z$, it follows that $y \succ^F_\P z$;
\item $\overline{C \cup D} \cup (C \setminus D) \cup (D \setminus C) = \overline{C \cap D}$, which is also decisive by claim ii). So, as for all $i \in \overline{C \cap D}$ $z \pp_i x$, it follows that $z \succ^F_\P x$.
\end{itemize}
The above forces the social preference $\p^F_\P$ to be cyclical, which is impossible by the definition of SPFs \eqref{eq:SPF}. Contradiction.
\end{enumerate}
All claims have been proven and the proof is therefore complete.
\end{proof}

\medskip

The second lemma concerns a general well-known fact about finite ultrafilters. Worth noticing is that the fact and its proof do not involve any reference to SPFs.
\begin{lemma}[Existence of a dictator] \label{lemma:dictators}
Let $\Win$ be an ultrafilter on $\N$. Then $\Win$ is {\em principal}, i.e.:
$
\exists i \in \N \ \mbox{s.t.} \ \set{i} \in \Win.
$
\end{lemma}
\begin{proof} 
Recall the definition of ultrafilter given in Lemma \ref{lemma:ultrafilter}. 
Since the set of voters $\N$ is finite by assumption, the closure $\bigcap \Win$ of $\Win$ under finite intersections belongs to $\Win$ by property {\em iii)}. 
We therefore have that $\bigcap \Win \neq \emptyset$. For suppose not, then $\N \not\in \Win$ by property {\em ii)}, against property {\em i)}. 
So, w.l.o.g., assume $i \in \bigcap \Win$ for $i \in \N$. We want to show that $\set{i} \in \Win$. Suppose towards a contradiction that $\set{i} \not\in \Win$. By property {\em ii)} we have that $\overline{\set{i}} \in \Win$, from which follows that $i \not\in \bigcap \Win$. Contradiction. Hence $\set{i} \in \Win$.
\end{proof}

\subsubsection{Completing the proof}

We can now pull the above lemmas together and prove the result we are after:

\begin{proof}[Proof of Theorem \ref{th:arrow}]
\RtoL Assume $F$ is a dictatorship. For each profile where $\N^{xy} = \N$, the dictator ranks $x$ over $y$. Hence $xy \in F(\P)$, proving Pareto. For any two profiles agreeing on the set of voters ranking $x$ over $y$ either the dictator belongs to that set, and therefore $xy$ belongs to both social preferences, or it does not and therefore $yx$ belongs to both social preferences. This proves IIA.

\LtoR By Lemma \ref{lemma:ultrafilter} the set of decisive coalitions under $F$ is an ultrafilter. By Lemma \ref{lemma:dictators} such ultrafilter is principal and therefore it contains a singleton. Such singleton is a decisive coalition and therefore, by Definition \ref{def:decisive}, the voter in the singleton is a dictator.
\end{proof}

%%%%%%%%%%%%%%%%%%%%%%%%%%%%%%%%%%%%%%%%%%%%

\section{Social choice by maximum likelihood \& closest consensus}

In this section we look at two other perspectives on social choice, and some of the voting rules they inspire: the maximum likelihood estimator (MLE) approach, which we already saw at work in the two alternatives case with the Condorcet jury theorem (Theorem \ref{th:condorcet}); and the consensus-based approach.

According to the MLE perspective a `true' (correct, objective) ranking of the alternatives in $A$ exists, and individual ballots are noisy estimates of such ranking. The voting rule should then identify the most likely ranking that has generated the observed profile, that is, it should be a {\em maximum likelihood estimator}. According to the consensus-based approach, the social preference is a form of `consensus' intended as a preference that is `as close as possible' to the observed profile. So the voting rule should identify the {\em closest consensus}.

\subsection{The Condorcet model when $m > 2$}

The model Condorcet developed in \cite{Condorcet1785} for the case in which $m > 2$ is based on the following assumptions (cf. Section \ref{sec:mle}):
\begin{description}
\item[C0] All rankings are equally likely a priory.
\item[C1] In each pairwise comparison each voter chooses the objectively better alternative with probability $p \in (0.5, 1]$.
\item[C2] Each voter's opinion on a pairwise comparison is independent of her opinion on any other pairwise comparison.
\item[C3] Each voter's opinions are independent of any other voter's opinions.
\item[C4] Each voter's judgment defines a linear order.
\end{description}
The four assumptions are clearly incompatible (specifically, {\bf C2} and {\bf C4}). Here, following \cite{young88condorcet}, we drop {\bf C4} and explore to what kind of voting rules assumptions {\bf C0}-{\bf C3} lead.
It is worth stressing that even though there are intuitive reasons against dropping condition {\bf C4} (after all it feels odd to allow for intransitive opinions), there are also natural reasons in favor of such a choice: if we take individual votes as the voter's best approximation at trying to identify the correct ranking, such votes may well be intransitive since best approximations of linear orders may well fail to be transitive (cf. \cite{truchnon08borda}).\footnote{The alternative choice consisting of retaining {\bf C4} while dropping {\bf C2} requires different noise models for the random generation of linear orders. One such widely deployed model is the so-called {\em Mallows noise model} \cite{mallows57non}.}

\subsection{Voting rules as MLE when $m>2$}

First, under the above assumptions {\bf C0}-{\bf C3} we are interested in finding a voting rule that identifies the most likely social preference. Two adjustments to the notion of SPF are needed. First, individual ballots may now not be linear orders as their strict part may contain cycles. They are, instead, tournaments (cf. Definition \ref{def:tournament}), that is, asymmetric, complete but not necessarily transitive binary relations over $\A$. We denote the class of tournaments over $A$ by $\Tour(\A)$. Second, there might be more than one most likely social preference. We therefore generalize the notion of SPF as follows. A generalized social preference function (GSPF) is a function:
\begin{align}
G: \Tour(\A)^n \to 2^{\LO(\A)} \label{eq:GSPF}
\end{align}
That is, it takes a vector of tournaments as input and outputs a set of linear orders. Notice that an SPF \eqref{eq:SPF}, when its output is a linear order (cf. Lemma \ref{prop:linear}), can therefore be viewed as a GSPF that always outputs a singleton set of linear orders.

The MLE of the correct ranking is therefore a GSPF $G$ such that, for every profile $\P \in \Tour(\A)^n$, $G(\P)$ is the set of linear orders $\p$ that maximize the probability of the observed profile $\P$ given $\p$. That is: 
\begin{align}
G(\P) & = \argmax_{\p \in \LO(\A)}(\Pr(\P \mid \p)).
\end{align}

We discuss a concrete application of this approach, and a slight variant of it.

\subsubsection{The Kemeny rule}

We first introduce one more piece of notation. The distance between two linear orders $\p$ and $\p'$, called {\em swap distance}, is defined as:\footnote{The distance is also known as, among others, Kendall tau or bubble-sort distance.}
\begin{align}
\swap(\p,\p') & = \size{\set{xy \in \A^2 \mid  x \pp y \ \mbox{and} \ y \pp' x}} \label{eq:swap}
\end{align}
That is, the swap distance between two linear orders is given by the number of pairs of alternatives with respect to which the two orders disagree. Equivalently $\swap(\p,\p')$ is the minimum number of swaps (i.e., inversions) of adjacent alternatives whereby one can obtain $\p'$ from $\p$.

\begin{remark}
Many alternative notions of distance can be defined. The simplest one is possibly the so-called {\em discrete distance}:
\begin{align} \label{eq:discr}
\discr(\p,\p') & =
\left\{
\begin{array}{ll}
1 & \mbox{if} \ \p \neq \p' \\
0 & \mbox{otherwise}
\end{array}
\right.
\end{align}
That is, the distance is $1$ for all and only the ballots that differ from $\p$. Another way to interpret it is to say that, if the cost of any operation on a linear order (e.g., swapping alternatives in arbitrary positions) is $1$, then in order to turn $\p$ into $\p'$ one incurs a cost of $0$ (if they are the same order) or of $1$ (otherwise).

In general, a distance is a function $d$ satisfying the following properties, for any $\p,\p',\p''$:
\begin{description}
\item{Non-negativity}, $d(\p, \p') \geq 0$;
\item{Identity of indiscernibles}, $d(\p,\p')=0$ iff $\p = \p'$;
\item{Symmetry}, $d(\p, \p') = d(\p',\p)$;
\item{Triangle inequality}, $d(\p, \p') \leq d(\p, \p'') + d(\p'',\p')$. 
\end{description}
\end{remark}

\begin{theorem} \label{th:kemeny_MLE}
Let $\P = \tuple{\p_1, \ldots, \p_n} \in \Tour(\A)^n$ be randomly generated according to assumptions {\bf C0}-{\bf C3}  above. Then
\[
\argmin_{\p \in \LO(A)}  \sum_{i \in \N} \swap(\p_i, \p)
\]
is the MLE of the correct ranking. 
\end{theorem}
\begin{proof}
For ease of presentation we work with the asymmetric (strict) part of linear orders.
Let $\pp \in \LO(\A)$ be the correct ranking. Let now $\pp' \in \Tour(\A)$ be a tournament that agrees with $\pp$ on $k$ pairs of alternatives. Observe that by \eqref{eq:swap} we have $\swap(\pp',\pp) = \binom{m}{2} - k$. So, by assumptions {\bf C1} and {\bf C2} the probability that a voter expresses ballot $\pp'$ is
\begin{align*}
p^k \cdot (1-p)^{\binom{m}{2}-k}  ~=~ p^{\binom{m}{2} - \swap(\pp',\pp)} \cdot (1-p)^{\swap(\pp',\pp)} ~=~ p^{\binom{m}{2}} \cdot \left( \frac{p}{1-p}\right)^{-\swap(\pp'_i,\pp)}.
\end{align*}
Then by assumption {\bf C3} we have that the probability of a profile $\tuple{\pp'_1, \ldots, \pp'_n}$, given that $\pp$ is the correct ranking, is proportional to
\begin{align*}
\prod_{i \in \N} \left( \frac{p}{1-p}\right)^{-\swap(\pp'_i,\pp)} & = \left( \frac{p}{1-p} \right)^{- \sum_{i \in \N} \swap(\pp'_i,\pp)}.
\end{align*}
By {\bf C0} the rankings that are most likely to be the correct one are the ones that maximize the probability of the observed profile. In our case, by {\bf C2} (since $p \in (0.5, 1]$, $ \frac{p}{1-p} > 0.5$), those are the rankings that minimize $\sum_{i \in \N} \swap(\pp'_i,\pp)$, thus proving the claim.
\end{proof}

The theorem provides an MLE justification of the following voting rule, which is known as the Kemeny rule \cite{kemeny59mathematics}.
\begin{VR}[Kemeny] \label{vr:kemeny}
The Kemeny rule is the SCF defined as follows. For all profiles $\P$:
\begin{align*}
\Kr(\P) & = \set{ \max_\p(A)  ~\middle|~ \p \in \argmin_{\p' \in \LO(\A)} \left( \sum_{i \in \N}\swap(\p_i, \p') \right)}
\end{align*}
Fixing a linear order $\p$, $\sum_{i \in \N}\swap(\p_i, \p)$ is called the {\em Kemeny distance} of $\p$ to $\P$.
\end{VR}
Intuitively, the Kemeny rule first identifies the linear orders which minimize the Kemeny distance of the order from the profile, then selects the top ranked alternatives in such linear orders.

\begin{example}[\cite{young95optimal}]
Let $n = 60$ and $m = 3$. Consider the following profile:
\[
\P =
\begin{array}{l | lll}
\# 23 & a & b & c \\
\# 17 & b & c & a \\
\# 2 & b & a & c \\
\# 10 & c & a & b \\
\# 8 & c & b & a \\
\end{array}
\]
Notice that the profile has no Condorcet winner. There are $3!$ possible linear orders. For each of them we can calculate their Kemeny distance from the above profile. For example, let us consider the order $bca$. We have:
\begin{align*}
% & 23 \cdot \swap(abc, abc) + 17 \cdot \swap(bca, abc) + 2 \cdot \swap(bac, abc) + 10 \cdot \swap(cab, abc) + 8 \cdot \swap(cba, abc)  \\
 %= & 0 + 34 + 4 + 20 + 24 = 82 \\
   &  23 \cdot \swap(abc, bca) + 17 \cdot \swap(bca, bca) + 2 \cdot \swap(bac, bca) + 10 \cdot \swap(cab, bca) + 8 \cdot \swap(cba, bca)  \\
=~& 46 + 0 + 2 + 20 + 8 \\
=~ & 76
\end{align*}
The Kemeny distance for the other rankings is computed in like fashion (left to the reader). The ranking $bca$ is in fact the one that minimizes the Kemeny distance. So $\Kr(\P) = \set{b}$.
\end{example}

\begin{fact} \label{fact:pl_k}
Assume $m = 2$. Then, for any profile $\P \in \LO(\A)^n$: $\Plr(\P) = \Kr(\P)$
\end{fact}
\begin{proof}
See Exercise \ref{ex:pl_k}.
\end{proof}
As a corollary of the above fact and Theorem \ref{th:kemeny_MLE} we have that, on two alternatives, the plurality rule (Rule \ref{vr:plurality}) is the MLE (cf. Theorem \ref{th:condorcet}) under the Condorcet's assumptions.

\subsubsection{The Borda rule as MLE}

We have shown that the Kemeny rule, a natural generalization of plurality, is the MLE under the assumptions {\bf C0}-{\bf C3}. By varying such assumptions it is possible to obtain similar MLE characterizations for different rules (cf. \cite{conitzer05common}). Here we show how this can be done for the Borda rule (Rule \ref{vr:borda}) by modifying in a specific way the error model (condition {\bf C2}) in order to require a specific probability for a voter to rank an alternative at position $k$ under the condition that such alternative is the top alternative in the correct ranking. In the case of Borda we are not interested in estimating the most likely true ranking, but the most likely true winner (i.e., the most likely top alternative in the correct ranking).\footnote{Technically this involves an MLE that is an SCF rather than a GSPF.} 

\begin{theorem}[Borda as MLE \cite{conitzer05common}] \label{th:borda_MLE}
Let $\P = \tuple{\p_1, \ldots, \p_n} \in \LO(\A)^n$ be randomly generated according to assumptions {\bf C0}, {\bf C1}, {\bf C3}, {\bf C4} and assuming that
\[
\Pr(~i(x) = k ~\mid~ \max_{\p}(\A) = x~) \sim 2^{m-k}.
\]
where $i(x)$ denotes the rank of $x$ in $i$'s ballot in $\P$.
Then the Borda rule is the MLE of the correct winner.
\end{theorem}
\begin{proof}
By the constraint on $\Pr(i(x) = k ~\mid~ \max_{\p}(\A) = x)$,  the probability of observing profile $\P$, given the true preference ranks $x$ as first, is therefore proportional to
\begin{align*}
\prod_{i = 1}^n 2^{m-i(x)} = 2^{\left( \sum_{i = 1}^n m-i(x) \right)}.
\end{align*}
Notice that $\sum_{i = 1}^n m-i(x) = \sum_{i \in \N} m-i(x)$ is precisely the Borda score (Rule \ref{vr:borda}). It follows that the alternative with the highest Borda score is also the alternative that is most likely the correct winner given the observed profile $\P$.
\end{proof}

%%%%%%%%%%%%%%%%%%%%%%%%%%%%%%%%%%%%%%%%%%%%

\subsection{Social choice as the closest consensus}

As briefly mentioned above the consensus-based approach to social choice is based on the following idea: we should first identify a compromise profile---a {\em consensus}---which is as close as possible to the observed profile, and than use that compromise profile to determine the social choice. 
% identifying a compromise `consensual' preference that is `as close as possible' to the observed profile. 
Depending on how the notions of `consensus' and `closeness' are made precise, different voting rules may be characterized. In this section we will discuss four prominent examples of characterizations of voting rules in terms of types of consensus and distance. 

\medskip

As a first example we mention again the Kemeny rule (Rule \ref{vr:kemeny}). We can think of the Kemeny rule as follows. First we identify a preference that is as close as possible---by swap distance---to the observed profile. That preference is then taken up by every agent creating a strongly unanimous profile where everybody agrees on the relative position of all alternatives. We can then take the top-ranked alternative in that unanimous preference to determine the social choice. So the rule can be reformulated as follows:
\begin{align}
\Kr(\P) & = \set{ \max_{\P'_i}(A)  ~\middle|~ i \in \N \ \mbox{and} \ \P' \in \argmin_{\P'' \in \mathcal{S}} \left( \sum_{j \in \N}\swap(\P_j, \P''_j) \right)}
\end{align}
where $\mathcal{S}$ is the class of profiles $\P \in \LO(\A)^n$ such that $\forall i, j \in \N$ $\p_i = \p_j$ (strong unanimity profiles). So we can justify the Kemeny rule both from an MLE and a consensus-based perspective. Recall that $\P_i$ denotes the $i$'th projection of $\P$, i.e., the linear order $\p_i$ of $i$ in $\P$. 

\medskip

As a second simple example, we show that the plurality rule (Rule \ref{vr:plurality}) is the SCF that minimizes the discrete distance from a consensual profile where every voter ranks the same alternative first:
\begin{theorem} \label{th:pl_consensus}
For all $\P \in \LO(\A)^n$
\begin{align*}
\Plr(\P) = \set{ \max_{\P'_i}(A)  ~\middle|~ i \in \N \ \mbox{and} \ \P' \in \argmin_{\P'' \in \mathcal{U}} \left( \sum_{j \in \N}\discr(\P_j, \P''_j) \right)}
\end{align*}
where $\mathcal{U}$ is the class of profiles $\P \in \LO(\A)^n$ such that $\forall i, j \in \N$ $\max_{\p_i}(\A) = \max_{\p_j}(\A)$ (unanimity profiles). 
\end{theorem}
\begin{proof}
See Exercise \ref{ex:pl_consensus}.
\end{proof}
Notice furthermore that since $\P' \in \mathcal{U}$, the top-ranked alternative of each voter is the same. That is why the social choice is determined by the top-ranked alternative of any agent $i$. So the plurality rule is the rule that selects the alternatives that are unanimously top-ranked in profiles that minimize the total distance---measured by $\discr$---from the input profile.

\medskip

A similar characterization exists also for the Borda rule (Rule \ref{vr:borda}).
\begin{theorem} \label{th:borda_consensus}
For all $\P \in \LO(\A)^n$
\begin{align*}
\Br(\P) = \set{ \max_{\P'_i}(A)  ~\middle|~ i \in \N \ \mbox{and} \ \P' \in \argmin_{\P'' \in \mathcal{U}} \left( \sum_{j \in \N}\swap(\P_j, \P''_j) \right)}
\end{align*}
where $\mathcal{U}$ is the class of unanimity profiles. 
\end{theorem}
\begin{proof}
See Exercise \ref{ex:pl_consensus}.
\end{proof}
Intuitively, under this characterization the Borda rule is the rule that selects the alternatives that are unanimously top-ranked in profiles that minimize the total distance---measured this time by $\swap$---from the input profile.

\medskip

Finally, we introduce the last rule of this chapter, which was explicitly defined in terms of closest consensus, the Dodgson rule \cite{dodgson76method}.
\begin{VR}[Dodgson] \label{vr:dodgson}
The Dodgson rule $\Dodg$ is the SCF defined as follows. For all profiles $\P$:
\begin{align*}
\Dodg(\P) & = \set{ \Con(\P')  ~\middle|~  \P' \in \argmin_{\P'' \in \mathcal{C}} \left( \sum_{j \in \N}\swap(\P''_j, \P_j) \right)}
\end{align*}
where $\mathcal{C}$ is the class of profiles $\P \in \LO(\A)^n$ for which there exists a Condorcet winner (Condorcet profiles) and $\Con$ is the Condorcet rule (Rule \ref{vr:condorcet}).
\end{VR}
Intuitively the rule outputs the alternatives that are Condorcet winners in the Condorcet profiles that are closest---according to $\swap$---to the input profile. Clearly such rule is Condorcet-consistent by definition.

Table \ref{tab:consensus} recapitulates the characterizations of the rules dealt with in this section in terms of type of consensus and type of distance.

\begin{table}[t]
\label{tab:consensus}
\begin{center}
\begin{tabular}{l| c c}
Rules & Consensus class & Distance \\
\hline
Plurality & Unanimous & Discrete \\
Borda & Unanimous & Swap \\
Kemeny & Strongly unanimous  & Swap \\
Dodgson & Condorcet & Swap
\end{tabular}
\end{center}
\caption{Types of consensus classes and distances for the characterizations of Plurality, Borda, Kemeny and Dodgson.}
\end{table}

%%%%%%%%%%%%%%%%%%%%%%%%%%%%%%%%%%%%%%%%%%%%

\section{Chapter notes}

The bulk of this chapter is again based on the introductions to voting theory provided in \cite[Ch. 2]{comsoc_handbook}, \cite[Ch. 1]{taylor05social}, \cite{endriss11logic} and, in addition, \cite{taylor08mathematics}.

The proof of Arrow's theorem presented above is based on \cite{hansson76existence}. The first proof of Arrow's theorem based on the ultrafilter argument is due to \cite{fishburn70arrow} and \cite{kirman72arrow}. The theorem has been proved with a wealth of different arguments. See for instance \cite{barbera80pivotal,reny01arrow,geanakoplos05three}.

Various generalizations and variants of Theorem \ref{th:borda_MLE} are studied in \cite{conitzer05common}. For an overview of results about the MLE and consensus-based approaches to voting rules see \cite{elkind16rationalizations}.

%%%%%%%%%%%%%%%%%%%%%%%%%%%%%%%%%%%%%%%%%%%%

\section{Exercises}

%\begin{exercise}
%Prove Fact \ref{fact:warmup}.
%\end{exercise}

\begin{exercise}
Consider the following rule:
\begin{VR}[Symmetric Borda]
The symmetric Borda voting rule is the SCF defined as follows. For all $\P \in \LO(\A)^n$:
\[
\Br^{s}(\P) = \argmax_{x \in \A} \sum_{y \in \A} \net^{xy}_\P 
\]
\end{VR}
That is, the rule selects the alternatives that beat all other alternatives (in the corresponding majority graph) by the largest margin. Prove that for any profile $\P$:
$
\Br(\P) = \Br^s(\P).
$
\end{exercise}

\begin{exercise}
Come up with a ballot profile that has the property that all rules mentioned in Table \ref{tab:rules_axioms} output a different social choice. You may ignore the Copeland rule (Rule \ref{vr:copeland}).
\end{exercise}

\begin{table}[t]
\begin{center}
\begin{tabular}{r|cccccc}
              & Resoluteness & Unanimity & Condorcet-cons. & Independence & Monotonicity & Liberalism \\ 
              \hline
Dictatorship & \checkmark & \checkmark &   & \checkmark & \checkmark &  \\
Plurality & &  & &  &  \\
Condorcet & & \checkmark & \checkmark & \checkmark & \checkmark &  \\
Copeland & &  &  &  &  &  \\
Borda & & & & &  &  \\
STV & & & & & & \\
\hline
\end{tabular}
\end{center}
\caption{Table stating which rules satisfy which axioms (only cases for Dictatorship and the Condorcet rule are provided).}
\label{tab:rules_axioms}
\end{table}

\begin{exercise} \label{ex:rules_axioms}
Consider Table \ref{tab:rules_axioms}. Rows correspond to voting rules, and columns to axioms we introduced in this and the previous chapter.
Choose two out of the four rules of Plurality, Copeland, Borda, STV.  Determine whether the rules you chose satisfy (\checkmark) or do not satisfy (no mark) the axioms given in the columns of the table. Provide a proof for each such statement.
By doing so you are completing Table \ref{tab:rules_axioms}.
\end{exercise}

\begin{exercise} \label{ex:sen}
Provide a proof of Theorem \ref{th:sen}. \fbox{Hint} argue towards a contradiction and construct a profile in which liberalism and Pareto together would force an empty social choice, which we know is impossible by the definition of SCF.
\end{exercise}

\begin{exercise}
Prove that every ultrafilter $\mathcal{U}$ on a an arbitrary set $X$ (see the definition of ultrafilter provided in Lemma \ref{lemma:ultrafilter}) is closed under supersets. That is: if $Y \in \mathcal{U}$ and $Y \subseteq Z \subseteq X$ then $Z \in \mathcal{U}$. 
\end{exercise}

\begin{exercise}
Let $\A = \set{a,b}$, $n$ odd and let $\Plr^\star$ denote the SPF corresponding to the plurality rule we introduced as SCF (Rule \ref{vr:plurality}). Consider the set of decisive coalitions $\Win_{\Plr^\star}$ under $\Plr^\star$. For each of the three properties of ultrafilters (cf. Lemma \ref{lemma:ultrafilter}) prove that $\Win_{\Plr^\star}$ satisfies that property or provide a counter-example if it does not.
\end{exercise}

%%%%%%%%%%%%%%%%%%%%%%%%%%%%%

\begin{exercise}[Properties of the Kemeny rule]
Determine whether the Kemeny rule satisfies the properties in Table \ref{tab:rules_axioms}. Justify your answers.
\end{exercise}

\begin{exercise} \label{ex:pl_k}
Prove Fact \ref{fact:pl_k}.
\end{exercise}

\begin{exercise} \label{ex:pl_consensus}
Prove Theorem \ref{th:pl_consensus}.
\end{exercise}

\begin{exercise} \label{ex:borda_consensus}
Prove Theorem \ref{th:borda_consensus}.
\end{exercise}

\begin{exercise}
Determine whether the Kemeny and Dodgson rules satisfy neutrality (Definition \ref{def:equal}). Justify your answer. 

\end{exercise}

%%%%%%%%%%%%%%%%%%%%%%%%%%%%%%%%%%%%%%%%%%%%%%%%%%%%%%%%%%%%%%%

\chapter{Eliciting Truthful Ballots} \label{ch:truthful}

The previous chapters tacitly assumed that the ballots provided by the voters report their genuine individual preferences: the voting rule can access truthful information and thus establish an appropriate social choice, according to the specific logic driving the rule. However, voting rules may create incentives for voters to misrepresent their true preferences by submitting a manipulated ballot to the rule, and thereby steer the rule's outcome towards better---for them---alternatives.

%%%%%%%%%%%%%%%%%%%%%%%%%%%%%%%

\section{Preliminaries: strategic voting}

\subsection{Voting strategically}

We start with three examples illustrating how voters could manipulate the social choice in their favour by misrepresenting their true preferences to the voting rule.

\begin{example}[Roman senate, continued]
Consider again the profile discussed in Example \ref{ex:pliny}
\[
\begin{array}{l | lll}
\# 102 & a & b & c \\
\# 101 & b & a & c \\
\# 100 & c & b & a
\end{array}
\]
Pliny demanded a vote by plurality that would lead to the choice of $\set{a}$. For $100$ voters (bottom row) this amounts to the worst possible alternative being selected as social choice. They have an incentive to misrepresent their ballots by submitting $b$ as their top-ranked alternative rather than $c$. The plurality rule would then be applied to the profile
\[
\begin{array}{l | lll}
\# 102 & a & b & c \\
\# 101 & b & a & c \\
\# 100 & b & c & a
\end{array}
\]
leading to the choice of $\set{b}$.\footnote{This is what in fact happened in the Roman senate, as documented in, e.g., \cite{szpiro10numbers}.}
\end{example}

\begin{example}[\cite{zwicker16rintroduction}]
Consider the following profile for $\N = \set{1, \ldots 7}$ and $\A = \set{a, b, c, d, e}$:
\[
\P = 
\begin{array}{l | lllll}
\# 2 & a & b & c & d & e \\
\# 3 & d & e & b & c & a \\
\# 2 & e & c & a & d & b
\end{array}
\]
The Borda rule selects $\set{e}$ with a score of $17$. For $2$ voters---those with ballot $abcde$---this is the worst possible alternative. It suffices for one of them, let it be voter $1$, to modify his ballot to $dabce$ in order to have $d$ instead selected as Borda winner. The new profile
\[
\P^\man = 
\begin{array}{l | lllll}
\# 1 & a & b & c & d & e \\
\# 1 & d & a & b & c & e \\
\# 3 & d & e & b & c & a \\
\# 2 & e & c & a & d & b 
\end{array}
\]
is such that $\Br(\P^\man) \pp_1 \Br(\P)$.
\end{example}

\begin{example}[\cite{conitzer16barriers}] \label{ex:SP_resolute}
Let now $\N = \set{1, 2, 3}$ and $\A = \set{a, b, c}$, and suppose the social choice is carried out using plurality with an alphabetic tie-breaking rule ($a > b > c$). Consider now the following ballot profile:
\[
\P = 
\begin{array}{l |lll}
1 & a & b & c  \\
2 & b & a & c  \\
3 & c  & b & a
\end{array}
\]
%where voter $3$ submits the top-ranked alternative according to her preference $c \pp_3 b \pp_3 a$. 
In this case the social choice would be $\set{a}$, making agent $3$ the least happy with the choice since $a$ is ranked at the bottom of her preference. However, were $3$ to misrepresent her ballot by reporting  $b \pp_3 c \pp_3 a$ instead, $b$ would become the social choice, granting her a better outcome.
\end{example}

\subsection{Manipulability and strategy-proofness}

The examples above show: 
that the Borda rule (Rule \ref{vr:borda}), in profiles where a single winner is selected, is manipulable by a single voter;\footnote{The French Academy of Sciences, of which Borda was a member, did experiment with the Borda rule but did not pursue its use further exactly because of this susceptibility to manipulation. Borda famously responded to this criticism by stating ``My scheme is intended for only honest men" \cite{black58theory}.} 
that the plurality rule (Rule \ref{vr:plurality}) can be manipulated by a coalition of voters, or by a single voter when it is supplemented by a deterministic tie-breaker (that enforces resoluteness). 

We can make all the above forms of manipulation of a social choice process mathematically precise. In this chapter, however, we will focus on single-voter manipulability of resolute rules as illustrated in Example \ref{ex:SP_resolute}. 
\begin{definition} \label{def:manipulability}
A resolute SCF $f$ (for a given $\tuple{\N, \A}$) is (single-voter) {\em manipulable} if there exist two profiles 
\begin{align*}
\P = \tuple{\p_1, \ldots, \p_i, \ldots \p_n} & \ \mbox{and} \ \P^\man = \tuple{\p_1, \ldots, \p_{i-1}, \p^\man_i, \p_{i + 1}, \ldots, \p_n} 
\end{align*}
such that $f(\P^\man) \pp_i f(\P)$. We say then that $i$ is a {\em manipulator}, that $\p_i$ is $i$'s {\em truthful} ballot and $\p^\man_i$ is $i$'s {\em untruthful} (or manipulated, or strategic) ballot.\footnote{Notice that by writing $f(\P^\man) \pp_i f(\P)$ we slightly abuse notation as the output of $f$ is, technically, a singleton.}
\end{definition}
Intuitively, an SCF is said to be manipulable whenever there exist situations (profiles) in which some voter has an incentive to submit an untruthful ballot to the SCF, that is, by manipulating his balllot the social choice is an alternative he prefers over the alternative that would be selected were he to vote truthfully.

It is worth noticing that the definition makes a full information assumption on the manipulator: the manipulator needs to know the ballots of all other agents in order to be able to select the right manipulated ballot. This is in general unrealistic, but it is a reasonable assumption when $n$ is small (e.g., deliberative committees) or when accurate enough information is available in the form of opinion pools (e.g., political elections). The assumption can also be interpreted as a worst-case assumption: as we cannot anticipate the level of information of a potential manipulator, we should conservatively assume full information.

\begin{remark}[Manipulability for irresolute SCFs]
We defined manipulability only for resolute rules. This is a simplifying---and arguably unrealistic---assumption, although we saw that resoluteness is a by-product of natural requirements such as IIA and Pareto (Proposition \ref{prop:linear}). It is however possible to define manipulability also in the context of irresolute SCFs. In such a case we need to specify how voters may rank sets of alternatives. Two natural definitions have been explored (cf. \cite{taylor05social}):
\begin{itemize}
\item For an {\em optimistic} manipulator $i$, $f(\P) \pp_i f(\P')$ whenever $\max_{\p_i}(f(\P)) \pp_i \max_{\p_i}(f(\P'))$;
\item For a {\em pessimistic} manipulator $i$, $f(\P) \pp_i f(\P')$ whenever $\min_{\p_i}(f(\P)) \pp_i \min_{\p_i}(f(\P'))$.
\end{itemize}
\end{remark}

\medskip

The problem with manipulability is that it creates the possibility for the process of social choice to be mislead as profiles may no longer represent the true preferences of the voters. In addition, manipulability creates indirect incentives for voters to invest energy and time into anticipating each others' manipulative behavior, rather than investing them on the substance of the decision at issue: if $i$ knows that $j$ has an incentive to manipulate, she will try to adapt her ballot accordingly, at which point $j$ may need to reconsider her ballot again, and so on. This motivates the following axiom.

\begin{definition} \label{def:strategyproof}
Let $\tuple{\N, \A}$ be given. A resolute SCF $F$ is {\bf strategy-proof} iff it is not manipulable.
\end{definition}
That is, no voter ever has an incentive to misrepresent her ballot.\footnote{In game-theoretic terms we say that reporting a truthful ballot is a dominant strategy in the voting game $\tuple{\N, \LO(\A)_i, f, \p_i}$.}

\subsection{Strategy-proofness and responsiveness axioms}

Strategy-proofness in resolute SCFs is directly linked to the responsiveness properties introduced in Definition \ref{def:resp_axioms}. In particular:
\begin{lemma}\label{prop:sp_mono}
Let $\tuple{\N, \A}$ be given. Let $f$ be a resolute and strategy-proof SCF. Then:
\begin{enumerate}[(a)]
\item $f$ is monotonic;
\item $f$ is independent;
\item if $f$ is non-imposed, it is Pareto. 
\end{enumerate}
\end{lemma}
\begin{proof}
\begin{enumerate}[(a)]

\item  We prove the claim by contraposition. So assume $f$ is not monotonic (Definition \ref{def:resp_axioms}). Then there exist two distinct alternatives $x \neq y \in \A$ and two profiles $\P, \P' \in \LO(\A)^n$ such that:
\begin{itemize}
\item $\N^{xy}_\P \subseteq \N^{xy}_{\P'}$, for all $y \in \A \setminus \set{x}$;
\item $\N^{yz}_\P = \N^{yz}_{\P'}$, for all $y, z \in \A \setminus \set{x}$;
\item $f(\P) = \set{x}$ and $x \not\in f(\P')$.
\end{itemize}
That is, in $\P'$ some more voters rank $x$ over $y$ while keeping the rankings of all other alternatives as in $\P$, and shifting the winner to a different alternative from $x$.
Let then $f(\P') = \set{z}$ with $z \neq x$. Notice there may be many such voters so there exists a sequence of intermediate profiles  $\P = \hat{\P}_1, \ldots, \hat{\P}_k = \P'$ differing only in the ballot of one voter and such that at some point $1 < \ell \leq k$ we have that $f(\hat{\P}_{\ell - 1}) = \set{x}$ and $f(\hat{\P}_{\ell}) = \set{z}$. So w.l.o.g. let us assume that $\P$ and $\P'$ are such profiles differing only in the ballot of one voter $i$ who has raised the ranking of $x$ w.r.t. $y$ in her ballot. Observe that $i \in \N^{yx}_\P \cap \N^{xy}_{\P'}$. There are two cases.
\fbox{$x \pp_i z$} Then in $\P'$ (where still $x \pp'_i z$ by how $\P'$ is constructed) $i$ can manipulate by submitting $\p_i$.
\fbox{$z \pp_i x$} Then in $\P$ $i$ can manipulate by submitting $\p'_i$.
In both cases $f$ is therefore manipulable, which proves the claim.
%Now if $i$'s truthful ballot is $y \pp_i x$ (as in $\P$) then $f(\P) \p_i f(\P')$. Therefore $f$ is manipulable by Definition \ref{def:manipulability}.

\item We prove the claim by contraposition. So assume $f$ is not independent (Definition \ref{def:more}). Then there exist two profiles $\P, \P' \in \LO(\A)^n$ and two alternatives $x, y \in \A$, such that $\N^{xy}_\P = \N^{xy}_{\P'}$, $x \in f(\P)$ (therefore $y \not\in f(\P)$ by resoluteness) and $y \in f(\P')$ (therefore $f(\P') = \set{y}$ by resoluteness). It follows that: there exists $i \in \N$ such that $\p_i \neq \p'_i$. By an argument analogous to the one provided in (a) we can focus w.l.o.g. to profiles $\P$ and $\P'$ that differ only in the ballot of $i$. There are two cases.
\fbox{$x \pp_i y$} Then in $\P'$ (where still $x \pp'_i y$ by how the profile is constructed) $i$ can manipulate by submitting $\p_i$.
\fbox{$y \pp_i x$} Then in $\P$ $i$ can manipulate by submitting $\p'_i$.
In both cases $f$ is therefore manipulable, which proves the claim.

\item See Exercise \ref{ex:sp_pareto}.
%We prove the claim by contraposition. So assume $f$ does not satisfy Pareto, then there exists a profile $\P$ where $\N^{xy}_\P = \N$ (therefore $\N^{yx}_\P = \emptyset$) and $f(\P) = \set{y}$. By monotonicity (ii) any profile $\P'$ in %which $\N^{yx}_\P \supset \emptyset$ is such that $f(\P) = \set{y}$. It follows that no profile exists for which $x$ is the winner, against the assumption of non-imposition. 

\end{enumerate}

\noindent
All items have been proven and the proof is complete.
\end{proof}
Any failure of monotonicity or independence in resolute SCFs generates a chance for manipulating behavior. In addition, monotonicity guarantees Pareto in the presence of non-imposition.

%%%%%%%%%%%%%%%%%%%%%%%%%%%%%%%

\section{There is no obvious strategy-proof social choice when $m > 2$}

The question we are after in this section then is: are there `nice' SCFs that are strategy-proof? Two immediate SCFs come to mind. First, dictatorships are trivially strategy proof as all ballots, except the dicator's, are irrelevant to determine the social choice. Second, simple majority---that is, plurality when $m=2$---is also strategy-proof: reporting a manipulated ballot inevitably lowers the plurality score of the top alternative.

\subsection{Strategy-proofness when $m = 2$}

For social choice contexts with two alternatives we can prove this simple consequence of May's theorem:

\begin{theorem}
Let $\tuple{\N, \A}$ be given such that $m = 2$ and $n$ is odd. A resolute SCF is anonymous, neutral and strategy-proof if and only if it is plurality.
\end{theorem}
\begin{proof}
\RtoL Straightforward.
\LtoR It follows from Theorem \ref{th:may} and Lemma \ref{prop:sp_mono} (monotonicity).
\end{proof}

\subsection{Strategy-proofness when $m > 2$}

We prove a fundamental theorem---akin to Arrow's (Theorem \ref{th:arrow})---showing that there exists no resolute SCF that is at the same time non-imposed, strategy-proof and non-dictatorial.
\begin{theorem}[Gibbard-Satterthwaite theorem \cite{gibbard73manipulation,satterthwaite75strategy}] \label{th:gib_satt}
Let $f$ be a resolute SCF for a $\tuple{\N, \A}$ s.t. $m > 2$. $f$ is non-imposed and strategy-proof if and only if it is a dictatorship.
\end{theorem}

\subsubsection{Proof strategy and blocking coalitions}

The proof we present relies on a concept related to the concept of decisive coalitions in the proof of Arrow's theorem: blocking coalitions.

\begin{definition}[Blocking coalition] \label{def:blocking}
Let $f$ be a resolute SCF (for a given $\tuple{\N, \A}$), and $x, y \in \A$. A coalition $C \subseteq \N$ is {\em blocking for} $y$ {\em by} $x$, under $f$, if
\[
\forall \P \in \LO(\A)^n:  \ \mbox{if} \ C \subseteq \N^{xy}_\P \ \mbox{then} \ y \not\in f(\P).
\] 
A coalition $C \subseteq \N$ is blocking if it is blocking w.r.t. every pair of alternatives. The set of all blocking coalitions is denoted $\Block$.
The set of blocking coalitions for $yx$ (under $f$) is denoted $\Block^{yx}_f$ (or simply $\Block^{yx}$ when $f$ is clear from the context).
\end{definition}
Observe that if $\Block$ contains a singleton then there exists a voter $i$ who is blocking for every pair of alternatives. Therefore, for any profile $\P \in \LO(\A)^n$, $f(\P) = \set{\max_{\p_i}(\A)}$, that is, $i$ is a dictator (Definition \ref{def:equal}).

The proof then relies on Lemma \ref{prop:sp_mono} and two further lemmas. The first one is an ultrafilter lemma for the set of blocking coalitions: any resolute SCF that is strategy-proof and non-imposed induces a set of blocking coalitions that takes the form of an ultrafilter. The second lemma is Lemma \ref{lemma:dictators}, which we used also in the proof of Arrow's theorem in the previous chapter.

\subsubsection{The ultrafilter lemma for resolute, non-imposed, strategy-proof SCFs}

\begin{lemma}[Ultrafilter lemma] \label{lemma:ultrafilter2}
Let $f$ be a resolute SCF (for a given $\tuple{N,A}$), that satisfies non-imposition and strategy-proofness.
The set $\Block$ of blocking coalitions (for $f$) is an {\em ultrafilter} over $\N$, that is:
\begin{enumerate}[i)]
%\item $\tuple{\Win, \subseteq}$ is a poset;

\item $\N \in \Block$, i.e., the set of all individuals is a decisive coalition;

\item $C \in \Block$ iff $\overline{C} \not\in \Block$, i.e., a coalition is decisive if and only if its complement is not;

%\item if $C \in \Win$ and $C \subseteq C'$ then $C' \in \Win$, i.e., if a coalition is decisive, all coalitions containing it are also decisive;

\item $\Win$ is closed under finite intersections: if $C, C' \in \Block$ then $C \cap C' \in \Block$, i.e., if two coalitions are winning then the individuals they have in common form a winning coalition.
\end{enumerate}
\end{lemma}

\begin{proof}
\begin{enumerate}[i)]

\item By non-imposition, for any alternative $x$ there exists a profile $\P$ such that $f(\P) = \set{x}$. Now take such profile $\P$ and construct a profile $\P'$ such that $\N^{xy}_{\P'} = \N$ for all $y \neq x$. By 
strategy-proofness and Proposition \ref{prop:sp_mono} (monotonicity), $f(\P') = \set{x}$. Again by Proposition \ref{prop:sp_mono} (independence), it follows that for any unanimous profile $\P''$ on $x$, $y \not\in f(\P'')$ for any alternative $y \neq x$. That is, $\N \in \Block$.   

\item \LtoR Suppose, towards a contradiction, that $C, \overline{C} \in \Block$. Consider now a profile $\P$ where: for all $i \in N$ the top two alternatives in each $\p_i$ are either $x$ or $y$; and $C = \N_\P^{xy}$ and $\overline{C} = \N_\P^{yx}$. This profile must exist as SCFs admit any profile in $\LO(\A)^n$ as input, and be such that $x \not\in F(\P)$ and $y \not\in F(P)$. Hence there exists $z \in \A \setminus \set{x,y}$ such that $f(\P) = \set{z}$, against Pareto (Lemma \ref{prop:sp_mono}). 

\RtoL 
%See Exercise \ref{ex:block}.
Assume $\overline{C} \not\in \Block$. Then there exists a profile $\P$ and alternatives $x$ and $y$ such that $\overline{C} \subseteq \N^{xy}_\P$ and $f(\P) = y$. By Lemma \ref{prop:sp_mono} (independence) for every profile $\P'$ such that $\N^{xy}_\P = \N^{xy}_{\P'}$, $x \not\in f(\P')$. That is, $\N^{yx}_\P$ is a blocking coalition for $x$ by $y$. Observe that $\N^{yx}_\P \subseteq C$. Then by Lemma \ref{prop:sp_mono} (monotonicity) $C$ is also blocking for $x$ by $y$. By adapting the argument provided earlier for Lemma \ref{lemma:contagion} (contagion lemma) we can conclude that $C \in \Block$ (see Exercise \ref{ex:contagion}).

\item We proceed towards a contradiction and assume that $C, D \in \Block$ and $C \cap D \not\in \Block$. By the previous item, $\overline{C \cap D} \in \Block$. Construct a profile $\P$ with the following two features.
One, all ballots are such that the top 3 alternatives are either $x$, $y$ or $z$. Second, the relative rankings of the alternatives in $\set{x, y, z}$ are as follows:\footnote{Cf. the same item in Lemma \ref{lemma:ultrafilter}}
\[
\begin{array}{l | lll}
C \cap D & x & y & z \\
D \setminus C & y & z & x \\
C \setminus D & z & x & y \\
\overline{C \cup D} & z & y & x
\end{array} 
\]
We have that:
\begin{itemize}
\item $(C \cap D) \cup (C \setminus D) = C$, which is blocking by assumption. So, as for all $i \in C$ $x \p_i y$, it follows that $y \not\in f(\P)$;
\item $(C \cap D) \cup (D \setminus C) = D$, which is decisive by assumption. So, as for all $i \in D$ $y \p_i z$, it follows that $z \not\in f(\P)$;
\item $\overline{C \cup D} \cup (C \setminus D) \cup (D \setminus C) = \overline{C \cap D}$, which is also blocking by claim ii). So, as for all $i \in \overline{C \cap D}$ $z \p_i x$, it follows that $x \not\in f(\P)$.
\end{itemize}
Therefore $f(\P) \cap \set{x, y, z} = \emptyset$. However, for all $w \not\in \set{x, y, z}$ there exists $w' \in \set{x, y, z}$ such that $w' \p_i w$. By Proposition \ref{prop:sp_mono} (Pareto) then $w \not\in f(\P)$. It follows that $f(\P) = \emptyset$, which is impossible by the definition of SCF \eqref{eq:SCF}. Contradiction.
\end{enumerate}
All claims have been proven.
\end{proof}

\subsubsection{The proof}

\begin{proof}[Proof of Theorem \ref{th:gib_satt}]
\RtoL It is straightforward to prove that a dictatorship is non-imposed and strategy-proof.
\LtoR By Lemma \ref{lemma:ultrafilter} the set of blocking coalitions under $f$ is an ultrafilter. By Lemma \ref{lemma:dictators} such ultrafilter is principal and therefore it contains a singleton. Such singleton is a blocking coalition and therefore, by Definition \ref{def:blocking} and the resoluteness of $f$, for any profile the social choice must consist of the top-ranked alternative of the voter in the singleton. Therefore $f$ is a dictatorship.
\end{proof}

%%%%%%%%%%%%%%%%%%%%%%%%%%%%%%%%%%%%%%%%%%%

\section{Coping with manipulability} \label{sec:coping}

So strategy-proofness is unattainable in general. Yet there are many ways to deal with this inherent limitation of social choice functions. We overview some of them in this section.

\subsection{Strategy-proofness when individual preferences are `coherent'}

\begin{example}[after \cite{zwicker16rintroduction}] \label{ex:hikers}
Three friends want to go on a hike together. There are three types of hikes: short $s$, medium $m$ and long $l$. So the social choice context is $\tuple{\set{1,2,3}, \set{s, m, l}}$. They hold the following true preferences:
\[
\P = 
\begin{array}{l | lll}
1 & l & m & s \\
2 & s & m & l \\
3 & m & s & l 
\end{array} 
\]
Observe that these preferences are aligned with the natural ordering of the alternatives based on distance, from longest to shortest $lms$: an individual either prefers longer hikes over shorter ones; or vice versa; or she prefers medium distance hikes. 
Let us denote this linear order by $l \ppp m \ppp s$.
No individual ranks the extremes ($l$ and $s$) above the middle alternative ($m$). Diagrammatically, with the line corresponding to $1$, the dashed line to $2$ and the dotted line to $3$:

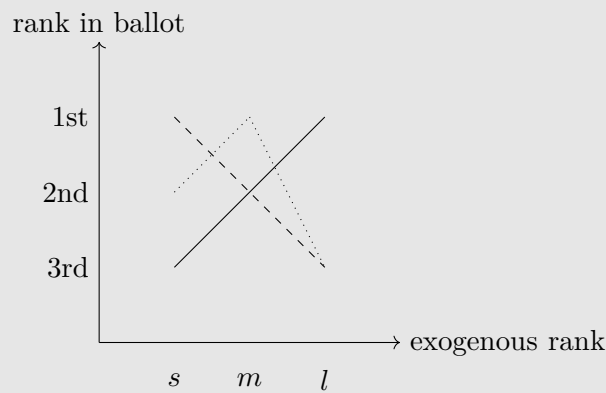
\begin{center}
\begin{tikzpicture}
 \draw [->] (0,0) -- (4,0) node [right] {exogenous rank}; 
 \draw [->] (0,0) -- (0,4) node [above] {rank in ballot}; 
 
 \draw (1,-0.5) node{$s$}; 
 \draw (2,-0.5) node{$m$};
 \draw (3,-0.5) node{$l$};   
 
\draw (0,1) node[left]{$3$rd}; 
\draw (0,2) node[left]{$2$nd}; 
\draw (0,3) node[left]{$1$st}; 

\draw (1,1) -- (3,3); 
\draw[dashed] (1,3) -- (3,1);  
\draw[dotted] (1,2) -- (2,3); 
\draw[dotted] (2,3) -- (3,1);

\end{tikzpicture}
\end{center}

There is a Condorcet winner in this profile: $m$.\footnote{So $\Con(\P) = \Cop(\P) = \set{m}$.} This guarantees the top-ranked alternative for $3$ and the second-best for $1$ and $2$. Furthermore, neither $1$ nor $2$ could obtain a better result by submitting a different ballot unless such ballot is not aligned any more with the natural ordering of the alternatives $l \ppp m \ppp s$.
\end{example}

\subsubsection{Single-peakedness as a form of preference coherence}

The example illustrates that when individual preferences are `similar enough', then manipulability ceases to be an issue. There are many ways in which this notion of `similarity' or `coherence' across individual preferences can be made precise. Here we present one among the most well-known, and the one that was first studied in the literature.

\begin{definition}[Single-peakedness \cite{black48rationale}]
A ballot $\p_i \in \LO(\A)$ is {\em single-peaked} if there exists a $\ppp \in \LO(\A)$ s.t., for all $x, y \in \A$:
\[
\mbox{if} \ \top_i \ppp x \ppp y \ \mbox{or} \ y \ppp x \ppp \top_i \ \mbox{then} \ x \pp_i y
\]
where $\top_i = \max_{\p_i}(\A)$ is called the {\em peak of} $i$ (in $\p_i$). We say then that $\p_i$ is single-peaked with respect to $\ppp$. The set of single-peaked linear orders over $\A$ is denoted $\SPLO(\A)$.

A profile of linear orders is {\em single-peaked} if there exists a $\ppp \in \LO(\A)$ s.t for all $i \in \N$, $\p_i$ is single-peaked with respect to $\ppp$. For a given profile $\P \in \SPLO(\A)^n$ we denote by ${\bf \top}_\P$ the vector of peaks of the ballots in $\P$ ordered by $\ppp$.
\end{definition}
Intuitively, a ballot is single-peaked whenever there exists an exogenous ranking on the alternatives such that whenever an alternative $x$ lies between the maximal of the ballot (the peak) and another alternatieve $y$, then the ballot ranks $x$ over $y$. A profile is single-peaked if all its ballots are single-peaked w.r.t. the {\em same} exogenous ranking. The exogenous ranking $\ppp$ represents a property (e.g., being left in the political spectrum) that the alternatives enjoy to an increasing or decreasing degree. The reader should check that profiles illustrating Condorcet's paradox (Example \ref{ex:condorcet}) are not single-peaked. Finally, observe that when $m = 2$ all profiles are trivially single-peaked.

\medskip

Single-peaked profiles, or profiles that are single-peaked with high probability (cf. \cite{regenwetter06behavioral}), are natural in many settings (e.g., political elections, facility locations, committee decision-making). Empirical studies have shown that single-peakedness is also sometimes the result of deliberative processes changing agents' preferences (cf. \cite{list13deliberation}). 

\subsubsection{Social choice on single-peaked preferences}

We want to see whether single-peakedness can give us strategy-proofness. First of all observe that an SCF defined on single-peaked profiles is a function
\begin{align}
f : \SPLO(\A)^n \to 2^\A \setminus \emptyset.
\end{align}

\begin{VR} \label{vr:median}
The median rule is the SCF defined as follows. For all $\P \in \SPLO(\A)^n$
\[
\median(\P) = \set{x \in \A \mid x \ \mbox{is median in} \ {\bf \top}_\P}
\]
\end{VR}
Observe that with the profile $\P$ of Example \ref{ex:hikers} $\median(\P) = \Con(\P) = \set{m}$. The Condorcet winner is the top-ranked alternative of the median voter.

\begin{theorem}[Black's theorem \cite{black48rationale}]
Let $\P \in \SPLO(\A)^n$. Then:
\begin{enumerate}[(a)]

\item If $n$ is even, the weak majority tournament $\p^\net_\P$ (Definition \ref{def:tournament}) is transitive and if $n$ is odd, the majority tournament $\pp^\net_\P$ is transitive.

\item If $n$ is odd then $\median(\P)$ is the singleton containing the Condorcet winner in $\P$. If $n$ is even then $\median(\P)$ is the set of weak Condorcet winners in $\P$.

\item For $n$ odd, $\median$ is strategy-proof. 

\end{enumerate}
\end{theorem}
\begin{proof}
\begin{enumerate}[(a)]

\item We have to consider two cases: $n$ even and $n$ odd. 
\fbox{$n$ odd} We want to prove that, for all distinct $x, y, z \in \A$, $x \pp^\net y$ and $y \pp^\net z$ implies $x \pp^\net z$. There are two cases to consider. 

\begin{description}

\item{\fbox{$B(xyz)$}} By assumption $x \pp^\net y$ we have that $\N^{xy} > \frac{n}{2}$. Since $y$ is between $x$ and $z$ in $\ppp$, by single-peakedness it follows that  for every $i \in \N^{xy}$, $y \pp_ i z$. 
From this, by the transitivity of individual preferences we obtain $\N^{xz} > \frac{n}{2}$ and therefore $x \pp^\net z$ as desired.
 
\item{\fbox{$B(yxz)$}} By assumption $x \pp^\net y$ we have that $\N^{xy} > \frac{n}{2}$. Now suppose towards a contradiction that $z \p^\net x$. By single-peakedness, all agents $i \in \N^{zx}$ would prefer $z$ over $y$ by transitivity, against the assumption that $y \pp^\net z$. We thus conclude that $x \pp^\net z$ as desired.

\item{\fbox{$B(xzy)$}} By assumption $x \pp^\net y$ and $y \pp^\net z$. By the above assumptions there exists an agent $i$ such that $x \pp_i y \pp_i z$. However, this preference is not single-peaked under the assumption that $z$ lies between $x$ and $y$.  As we assumed that $\P \in \SPLO(\A)^n$ this case is therefore impossible.

\end{description}
In all (possible) cases we conclude $x \pp^\net z$ as desired, proving the claim. \fbox{$n$ even} The argument is similar and needs to consider the possibility of split majorities in pairwise comparisons.

\item There are two cases. \fbox{$n$ odd} Then $\median(P)$ is a singleton $\set{\top_i}$. We need to show that for all $x \in \A$,  $\N^{\top_i, x} > 0$. Now let:
\[
{\bf \top}_\P = \tuple{\underbrace{\ldots}_L, \top_i, \underbrace{\ldots}_R}
\]
For all voters $j$ whose peaks appear in $L$ we have that $\top_i$ is preferred over any alternative $x$ to the right of $\top_i$ in $\ppp$ (and therefore to any peak in $R$). Vice versa, for all voters $j$ whose peaks appear in $R$ we have that $\top_i$ is preferred over any alternative $y$ to the left of $\top_i$ in $\ppp$ (and therefore to any peak in $L$). It follows that $\top_i$ beats any other alternative in a pairwise comparison, i.e., it is a Condorcet winner.

 \fbox{$n$ odd} Then $\median(\P)$ contains two elements $\set{\top_i, \top_j}$. Assume w.l.o.g. that $\top_i \ppp \top_j$. Now let:
\[
{\bf \top}_\P = \tuple{\underbrace{\ldots}_L, \top_i, \top_j \underbrace{\ldots}_R}
\]
We reason in a similar fashion as in the case for $n$ odd to conclude that $ \top_i$ and  $\top_j$ are weak Condorcet winners.

\item Let $\P$ be given, and let $i$ be the manipulator. Assume w.l.o.g. that $i$'s peak lies to the right (w.r.t. $\ppp$) of $\median(\P)$. She has two ways to manipulate her ballot by submitting an untruthful peak $\top'_i$. \fbox{$\top_i \ppp \top'_i$} Submit a ballot with peak $\top'_i$ further to the right of $\top_i$. In such a case $\median(\P)$ would not change. \fbox{$\top'_i \ppp \top_i$} Submit a ballot with peak $\top'_i$ to the left of $\top_i$. In such a case, if $\median(\P)$ changes, it changes to an alternative further away (w.r.t. $\ppp$) to $i$'s truthful peak. So no manipulation for $i$ exists.  
\end{enumerate}
The proof is complete.
\end{proof}

%Observe that a consequence of (a) is that the (strict) majority tournament $\pp^\net$, which arises for $n$ odd, is also transitive (see Exercise \ref{ex:trans}).

\medskip

In this section we showed how the issue of manipulability can be sidestepped when it is possible to make certain assumptions about the coherence of the voters' true preferences. If restricting the domain of SCFs in such a way is not possible, there are a number of other avenues that have been explored in order to manage the issue of manipulability. We turn to them now, just sketching the key ideas behind each approach.

%%%%%%%%%%%%%%%%%%%%%%%%%%%%

\subsection{Making manipulation difficult}

We sketch three ways in which manipulations, even though possible, may be considered difficult or unlikely: by making them computationally intractable, by making their informational requirements matter, and by introducing randomness in the SCFs. We sketch them in turn.

\subsubsection{Using computational complexity}

The key intuition of this approach is the following: even if an SCF is manipulable, finding the manipulating strategy may be an intractable computational problem. The task of finding a suitable manipulation can be approximated by considering the following decision problem:

\begin{description}
\item[Given] A partial profile $\P_{-i}$ without the vote of the manipulator $i$; and $i$'s preferred alternative $x$
\item[Question] Is there a ballot $\p_i$ that $i$ can submit so that $f(\P_{-i}, \p_i) = \set{x}$?
\end{description}

It has been shown (e.g., \cite{bartholdi89computational}) that the above decision problem is solvable in time polynomial  (in $n + m$) for many of the rules we discussed (Plurality, Borda) by using a simple greedy algorithm that works as follows:
\begin{description}
\item{Initialization:} Rank $x$ at the top;
\item{Repeat:} Check whether another alternative can be ranked immediately below the previous one without making $x$ lose; if that is the case, do so; if that is not possible, manipulation is impossible.  
\end{description}

However, the problem has been shown to be intractable (NP-hard) for some other common rules like, in particular, STV (Rule \ref{vr:STV}) \cite{bartholdi91single}.

\subsubsection{Using uncertainty}

Definition \ref{def:manipulability} can be modified in order to make it sensitive to the knowledge that is actually available to the manipulator.

\begin{definition}[Dominating manipulations \cite{conitzer11dominating}] \label{def:domman}
Fix a true profile $\P$ (for a given $\tuple{\N, \A}$).
Let $IS: \N \to 2^{\LO(\A)^n}$ assign to each voter a set of profiles (information set) such that $\P \in IS(i)$ for all $i \in \N$, and all profiles in $IS(i)$ agree on $\p_i$.
A resolute SCF $f$ is {\em vulnerable to dominating manipulations} (in $\P$) if there exist a voter $i$ and ballot $\p^\man_i$ s.t 
\begin{enumerate}[(a)]
\item $f(\P'_{-i}, \p_i^\man) \p_i f(\P')$ for all $\P' \in IS(i)$,
\item $f(\P'_{-i}, \p_i^\man) \pp_i f(\P')$ for some $\P' \in IS(i)$
\end{enumerate}
where $\P'_{-i}$ denotes the partial profile obtained from $\P'$ by removing $i$'s ballot. Ballot $\p^\man_i$ is called a {\em dominating manipulation} by $i$. An SCF is said to be {\em immune to dominating manipulations} iff it is not vulnerable to dominating manipulations in any $\P$.
\end{definition}
Observe that dominating manipulations reduce to simple manipulations (Definition \ref{def:manipulability}) whenever a manipulator has full information, that is, his information set is $\set{\P}$.
The definition allows then to break the negative result of Gibbard-Satterthwaite theorem by exploiting the fact that the larger the information set, the harder it becomes to find a dominating manipulation.

\begin{fact} \label{fact:dominating}
The Borda rule with a deterministic tie-breaking rule is immune to dominating manipulation when the manipulator's information set is $\subseteq$-maximal (i.e., the manipulator has no information).
\end{fact}
\begin{proof}
See Exercise \ref{ex:dominating}.
\end{proof}

\subsubsection{By randomization}

One way to rule out (deterministic) manipulations altogether is by introducing randomization in the workings of SCFs. A randomized SCF (RSCF) is a function
\begin{align}
r: \LO(\A)^n \to \Delta(\A)
\end{align}
where $\Delta(\A) = \set{{\bf x} \in \mathbb{R}^m \mid \forall i \in \A, {\bf x}_i \geq 0 \ \mbox{and} \ \sum_{i \in \A} {\bf x}_i = 1}$ is the set of all lotteries (distributions) over $\A$.

An elegant theorem by Gibbard  \cite{gibbard77manipulation} shows that the random draw of one agent (the so-called randomized dictatorship) from $\N$ is the only RSCF that is `lottery' strategy-proof (in the specific sense of stochastic dominance\footnote{A lottery $p$ stochastically dominates a lottery $q$ for $i$ iff  $\sum_{x \in A: x P_i y} p(x) \geq \sum_{x \in A: x P_i y} q(y)$, for any $y \in A$. A sortition rule is strategy-proof, w.r.t. stochastic dominance if for all $i \in \N$, it never selects a lottery which is stochastically dominated by another lottery for $i$.}
) 
and `lottery' efficient (in the sense of never assigning positive probability to alternatives that are Pareto dominated).

%%%%%%%%%%%%%%%%%%%%%%%%%%%%%%

\subsection{Letting manipulation be}

What happens if every voter is allowed to act as a manipulator? In particular, can we quantify exactly how bad the social choice on the fully manipulated profile would be when compared with the social choice supported by a truthful profile? To answer this question requires a game-theoretic perspective on voting.\footnote{For an introduction to game theory see  \cite{osborne94course}.}

\medskip

Any (resolute) SCF $f$ for a given context $\tuple{\N, \A}$ induces an ordinal game $\G = \tuple{\N, \A, \LO(\A)^n, \P, f}$ where:
\begin{itemize}
\item $\N$ is the set of players;
\item $\A$ is the set of outcomes;
\item $\LO(\A)$ is the strategy space of each agent (this could be coarsened to $A$, e.g., for plurality);
\item $\P$ are the (true) preferences of the players;
\item $f$ is the outcome function mapping strategy profiles to outcomes.
\end{itemize}

A (pure-strategy) Nash Equilibrium of $G$ is a profile $\P = \tuple{\p_1, \ldots, \p_n}$ such that there exists no player $i \in \N$ such that for some $\p'_i$
\[
f(\P_{-i}, \p'_i) \pp_i f(\P)
\]
That is, no player can profit by unilaterally misrepresenting her preference with a different ballot. In other words, no agent has a so-called {\em profitable deviation} in $\P$.
Another way to think of Nash Equilibria is via graphs over the set of profiles. The set $NE(G)$ of Nash equilibria of $G$ are also the sinks in the graph $\tuple{\LO(\A)^n, \to}$ where $\to$ links two profiles whenever the second can be obtained from the first by letting one agent execute a profitable deviation. 

\begin{example} \label{ex:dynamic}
Recall Example \ref{ex:SP_resolute} with $\N = \set{1, 2, 3}$ and $\A = \set{a, b, c}$, and $f$ be plurality with an alphabetic tie-breaking rule ($a > b > c$). Assume the following profile of true preferences
\[
\P = 
\begin{array}{l | lll}
1 & a & c & b \\
2 & b & c & a \\
3 & c & b & a 
\end{array} 
\]
From this profile agent $3$ has a best response consisting of ballot $bca$, leading to a NE where $\set{b}$ is the social choice. Similarly, agent $2$ has also a best response consisting of ballot $cba$, which leads to a NE where, instead, $\set{c}$ is the social choice.
\end{example}

The question that has been asked in the iterative voting literature \cite{branzei13how,meir18iterative} is how far a NE is from the truthful profile, once such NE is obtained through a sequence of best responses starting at the truthful profile. A way to make this question precise for the class of scoring rules (Definition \ref{def:scoring}) is by using the so-called (dynamic) price of anarchy:
\begin{align}
PoA(f) & = \min_{\P} \min_{\P' \in NE_\P} \frac{s_f(f(\P),\P)}{s_f(f(\P'),\P')}
\end{align}
where $s_f(x, \P)$ denotes the score of $x$ given $\P$ according to $f$, and $NE_\P$ is the set of NE reachable from $\P$ via best-response dynamics. That is, PoA denotes the worst case ratio between the true score of an alternative and the score it would receive in equilibrium.

It has been shown \cite{branzei13how} that $PoA(\Plr)$, with a deterministic tie-breaker, is close to $1$. In other words, full-blown manipulation does not have much impact on the quality of the outcome. A much more negative result has been proven for the Borda rule (cf. \cite{branzei13how}).

%%%%%%%%%%%%%%%%%%%%%%%%%%%%%%%%%%%%%%%%%%%%

\section{Chapter notes}

The bulk of this chapter is again based on the introductions to voting theory provided in \cite[Ch. 2]{comsoc_handbook}, \cite[Ch. 1]{taylor05social}, \cite{endriss11logic} and \cite{taylor08mathematics}.

The proof of Theorem \ref{th:gib_satt} is based on \cite{batteau81stability}. Since the original articles \cite{gibbard73manipulation,satterthwaite75strategy} many alternative proofs of the theorem have been presented, e.g.: \cite{barbera83strategy,benoit00gibbard}.

Section \ref{sec:coping} is based on several sources: \cite{conitzer16barriers} (complexity of manipulation), \cite{brandt18rolling} (randomization), \cite{meir18iterative} (voting games), \cite{conitzer11dominating} (information).

%%%%%%%%%%%%%%%%%%%%%%%%%%%%%%%%%%%%%%%%%%%%

\section{Exercises}

\begin{exercise} \label{ex:sp_pareto}
Prove item {\em iii)} of Lemma \ref{prop:sp_mono}.
\end{exercise}

%\begin{exercise} \label{ex:block}
%Prove item {\em ii)} \RtoL of Lemma \ref{lemma:ultrafilter2}.
%\end{exercise}

\begin{exercise}
Define a variant of manipulability for irresolute SCFs as follows: an SCF $f$ (for a given $\tuple{\N, \A}$) is {\em manipulable} if there exist two profiles 
\begin{align*}
\P = \tuple{\p_1, \ldots, \p_i, \ldots \p_n} & \ \mbox{and} \ \P^\man = \tuple{\p_1, \ldots, \p_{i-1}, \p^\man_i, \p_{i + 1}, \ldots, \p_n} 
\end{align*}
such that $f(\P^\man)$ and $f(\P)$ are singletons and $f(\P^\man) \pp_i f(\P)$.

Prove that the Plurality rule (Rule \ref{vr:plurality}) is not manipulable according to the above definition. 
\end{exercise}

\begin{exercise} \label{ex:contagion}
Let $F$ be a resolute SPF (for a given $\tuple{\N, \A}$) which is non-imposed and strategy-proof. Prove that if $C \in \Block^{xy}$ for some $x,y \in A$, then $C \in \Block$. This is the equivalent of the contagion lemma (Lemma \ref{lemma:contagion}) for the Gibbard-Sattertwhaite theorem (Theorem \ref{th:gib_satt}).
\end{exercise}

\begin{exercise}
In no more than three sentences discuss the differences between Arrow's theorem (Theorem \ref{th:arrow}) and the Gibbard-Sattertwhaite theorem (Theorem \ref{th:gib_satt}).
\end{exercise}

%\begin{exercise} \label{ex:trans}
%Prove that if a total relation $R$ is transitive, then its asymmetric part (that is, the edges of $R$ that are not symmetric) is also transitive.
%\end{exercise}

\begin{exercise}
For $\size{A} = m$, $\size{\LO(\A)^n}= (m !)^n$. Now fix an order $\ppp$. How many single-peaked profiles exist for $\ppp$?
%What is $\size{\SPLO(\A)^n}$? 
Justify your answer.
\end{exercise}

\begin{exercise}
Determine whether the median voter rule (Rule \ref{vr:median}) satisfies: unanimity, independence, monotonicity. Justify your answers.
\end{exercise}

\begin{exercise}
Use the variant of May's theorem with ties (Theorem \ref{th:may2}) to prove that, when $m = 2$, the median rule (Rule \ref{vr:median}) is anonymous, neutral and positively responsive. 
\end{exercise}

\begin{exercise} \label{ex:dominating}
Prove Fact \ref{fact:dominating}.
\end{exercise}

\begin{exercise}
Does the Gibbard-Satterthwaite theorem (Theorem \ref{th:gib_satt}) still hold if we drop the assumption of non-imposition? Justify your answer. If the theorem still holds explain how the proof can go through without that assumption. If the theorem does not hold any more, and therefore there exist SCFs which is strategy-proof and non-dictatorial, provide one of such functions.
\end{exercise}

%%%%%%%%%%%%%%%%%%%%%%%%%%%%%%%%%%%%%%%%%%%%%%%%%%%%%%%%%%%%%%%

\chapter{Choosing Many Out of Many}

The social choice problem we considered so far concerns how to select {\em one} `best' alternative given a profile of individual rankings, or more generally a set of tied `best' alternatives. In this chapter we present some results on the related problem of selecting one `best' set of alternatives---a so-called {\em committee}---or, more generally, one tied set of sets of `best' alternatives. Voting for committee selection is a much more recent, and therefore less consolidated, area of research in social choice and many directions of research in this area are still open.

\section{Preliminaries}

We follow the same approach used to introduce and discuss social choice functions: we provide the general definition of functions for committee selection; we then provide concrete examples of such functions; and we finally show how also these functions can be studied from an axiomatic point of view.

%\Comment{Say something about social choice correspondences: $f_k: \LO(\A)^n \to 2^A$}

\subsection{Multi-winner social choice}

We are interested in situations where a subset of alternatives of a given size $k \leq m$, the {\em committee}, is to be selected based on the preferences over $A$ of $n$ individuals. Given a context $\tuple{\N, \A}$ a multi-winner SCF or committee-selection function (CSF), for a given committee size $1 \leq k \leq n$, is a function
\begin{align}
f_k: \LO(\A)^n \to 2^{\set{X \subseteq \A ~\mid~ \size{X} = k}} \backslash \emptyset \label{eq:multi}
\end{align}
Given a profile of linear orders, the function outputs a set of sets of alternatives (i.e., a set of committees), all of size $k$. Obviously a CSF $f_1$ is equivalent to an SCF, as it selects a committee of size $1$. When we leave $k$ unspecified, we talk about a family of CSFs mapping every integer $k$ in $\set{1, \ldots n}$ to the CSF $f_k$. In this chapter, we refer to families of CSFs as committee selection  or multi-winner rules. 
%We will also refer to $f_k$ as a function $f: \set{1, \ldots n} \to (\LO(\A)^n \to \set{X \subseteq 2^\A \mid \size{X} = k})$ as a {\em family of} CSFs.

\begin{remark}
Notice that there is a natural way in the problem of committee selection can be casted as a social choice problem where committees of size $k$ {\em are} the alternatives relevant for the social choice. In other words the social choice context would be $\tuple{\N, \set{X \subseteq 2^\A \mid \size{X} = k}}$, and voters would need to submit ballots ranking all elements in such a set, which would therefore $\binom{m}{k}$ long. Given a profile of such ballot an SCF could be applied to select a tied set of committees. The type in \eqref{eq:multi} takes a different---and for obvious reasons more practical---approach selecting committees using only information consisting of rankings of alternatives.
\end{remark}

\subsection{Where are CSFs relevant}

Three application domains are normally mentioned for CSFs:
\begin{description}

\item[Excellence-driven] Here the issue is to select the best $k$ alternatives from $\A$. A typical example of this type of committee selection problem are shortlisting processes.

\item[Diversity-driven] Here the issue is to select $k$ alternatives that can cover all possible views on $\A$. Typical examples of this type are: facility location problems (e.g., where to place $k$ new schools), movie selection for an airline entertainment system, product selection for the homepage of an internet store.

\item[Proportionality-driven] Here the issue is to select $k$ alternatives that represent the views of $\N$ proportionally. The typical example are parliamentary elections.

\end{description}
All the above domains come intuitively with different requirements on how the CSF should ideally behave in order to select suitable committees. Such intuitive requirements underpin the axiomatic perspective on CSFs we present in Section \ref{sec:csf_axioms} below.

%%%%%%%%%%%%%%%%%%%%%%%%%%

\section{Some multi-winner voting rules}

For each of the application domains mentioned above we provide examples of common multi-winner rules. Finally we also provide examples of staged rules for committee selection which are commonly used in practice.

\subsection{Excellence-driven rules}

To introduce some of the rules in this class we need some auxiliary terminology. Best-$k$ multiwinner rules are based on the notion of social preference function (SPF) we encountered in Chapter \ref{ch:many-one}. Recall that these functions map profiles of linear orders (the individual preference) to a total preorder, that is, a ranking that may contain ties. Notice, however, that each total preorder $\preceq$ can be represented by a set of linear orders, that is, all the possible linear orders that one obtains by resolving the ties in $\preceq$. The intuition behind best-$k$ multiwinner rules is to form a committee by simply selecting the top $k$ alternatives in some of those linear orders representing possible social preferences. 

One more piece of notation: to the note the top $k$ element of a linear order $\p \in A^2$ we write $\max^k_{\p}(\A)$. We move now to the formal definition.

%Best-$k$ rules are based on a generalized form of the social preference functions (SPFs) we encountered in Chapter \ref{ch:many-one}, where the social preference is represented by a set of linear orders rather than a total preorder. So a %generalized social preference function (GSPF) is a function $G: \LO(\A)^n \to 2^{\LO(\A)}$ \eqref{eq:GSPF}. The intuition here is that individual preferences are mapped to a set of (tied) social preferences.

\begin{definition}[Best-$k$ rules]  \label{def:SnC}
Given an SPF $G$ a best-$k$ rule $\SnC^G$ (for $G$) is the family of CSFs defined as follows. For any profile $\P \in \LO(\A)^n$ and integer $1 \leq k \leq n$
\[
\SnC^G_k(\P) = \set{X \subseteq \A \mid \ \mbox{exists} \ \p \in G(\P), \mbox{s.t.} \ X = \max^k_{\p}(\A)}.
\]
A family of CSFs $f$ is said to be a best-$k$ rule whenever there exists a $G$ such that, for every $1 \leq k \leq n$, $f_k = \SnC^G_k(\P)$.
\end{definition}

Depending on the choice of $G$ in Definition \ref{vr:SnC} we obtain different multi-winner rules. In particular, if we take $G$ to be the GSPF associating to any profile the set of linear orders refining a weak order induced by an SPF $F$ (like a scoring function) we obtain, for example: 
\begin{VR} \label{vr:SnC}
The following are best-$k$ ($1 \leq k \leq n$) multiwinner rules:
\begin{description}

\item[$k$-Plurality], i.e., $\SnC^\Plr_k(\P)$, also known as single non-transferable vote

\item[$k$-Approval], i.e., $\SnC^{\Ar_k}_k(\P)$, also known as bloc voting

\item[$k$-Borda], i.e., $\SnC^\Br_k(\P)$

\end{description}
where $\Plr$, $\Ar$ and $\Br$ refer to the SPFs ranking alternatives by their plurality, approval and, respectively, Borda scores.
\end{VR}
One can in the same fashion use the Copeland (Rule \ref{vr:copeland}), Kemeny (Rule \ref{vr:kemeny}) or Dodgson (Rule \ref{vr:dodgson}) scores to obtain the corresponding rankings and then cut them to the top $k$ alternatives to obtain the committee. These are sometimes referred to also as `score-and-cut' and `rank-and-cut' rules.

\subsection{Diversity-driven rules}

Several diversity-driven multiwinner rules can be formulated as follows:

\begin{definition}[Chamberlin-Courant \cite{chamberlin83representative}] 
For an integer $1 \leq k \leq n$ and scoring vector $\w$ (Definition \ref{def:scoring}) a Chamberlin-Courant rule for size $k$ is defined as follows. For any profile $\P \in \LO(\A)^n$:
\[
\ChCo^{\w}_k(\P) = \argmax_{C, \size{C} = k} \sum_{x \in C} \sum_{i \in \N^{Cx}_\P} w_{i(x)}
%\ChCo_k(\P) = \argmin_{C \in 2_k^\A} \sum_{i \in \N} \size{x \in \A \mid x \pp_i r(i)}
\]
where $\N^{Cx}_\P = \set{i \in \N \mid x = \max_{\p_i}(C)}$ is the set of voters $i$ that rank $x$ as top among the alternatives in $C$ ($x$ is then said to be the {\em representative} of $i$ in $C$);  and $i(x)$ as usual denotes the position of $x$ in $\p_i$. Quantity $\sum_{x \in C} \sum_{i \in \N^{Cx}_\P} w_{i(x)}$ is called the {\em representativeness value} of a given committee $C$. 
\end{definition}

By varying the scoring vector $\w$ we thus obtain different families of CSFs. For example:
\begin{VR} \label{vr:chco}
The Borda-based $k$-Chamberlin-Courant multiwinner rule is the rule $\ChCo^{\w}_k$ where $\w$ is the Borda scoring vector.
\end{VR}
The intuition behind this rule is simple. It forms the committee that contains as many representatives as possible, and with the highest possible Borda score. So, if an alternative is to be added to the committee, alternatives with higher ranks from agents not already represented by the committee will be preferred over alternatives with lower ranks from agents already represented by the committee.

\subsection{Proportionality-driven rules}

A simple proportionality-driven rule based on approval voting is the following one:
\begin{VR} \label{vr:PAV}
The proportional $k$-approval voting rule is defined as follows ($1 \leq k \leq n$):
\[
\PAV_k(\P) = \argmax_{C, \size{C} = k} \sum_{i \in \N} \left( 1 + \frac{1}{2} + \ldots + \frac{1}{\size{\max^k_{\p_i} \cap~C}} \right)
%\ChCo^{\w}_k(\P) = \argmax_{C, \size{C} = k} \sum_{x \in C} \sum_{i \in \N^{Cx}_\P} w_{i(x)}
%\ChCo_k(\P) = \argmin_{C \in 2_k^\A} \sum_{i \in \N} \size{x \in \A \mid x \pp_i r(i)}
\]
\end{VR}
That is, the rule looks at the top $k$ candidates in each individual preference (cf. Rule \ref{vr:approval}) and tries to find the committee that maximizes a score that is computed by giving to each agent $i$ as many points as $1 + \frac{1}{2} + \ldots + \frac{1}{\ell}$, where $\ell$ is the number of alternatives in the committee that are approved by $i$. So, notice that the rule will always give priority to committees that contain a few alternatives from more agents rather than committees that contain many alternatives from fewer agents, like in the Chamberlin-Courant rules. At the same time, alternatives that are more often approved of will be more likely to be included in the committee, reflecting a form of proportionality.

\begin{remark}
The way we defined $\PAV_k$ is more restrictive than the way in which proportional approval voting is normally defined in the literature, where it is based on simple approval voting rather than $k$-approval. Cf. \cite{faliszewski17multiwinner}. Proportional approval voting was first proposed by the Danish matematician Thorvald Thiele at the end of the 19th century in \cite{thiele1895}, at the dawn of electoral suffrage in the Nordic countries. It is for this reason also known as {\em Thiele method}. The rule was used in practice at the beginning of the 20th century in Sweden. 
\end{remark}

\subsection{Choose-and-repeat rules}

\begin{VR}[Sequential plurality] \label{vr:splr}
The sequential plurality rule is the CSF defined as follows, for any profile $\P \in \LO(\A)^n$ and integer $k \geq 1$:
\begin{align*}
\SPlr_k(\P) & = \set{X \mid \mbox{exist} \ x^1, \ldots, x^k \ \mbox{s.t.} \  X = \set{x^1, \ldots, x^k} }
\end{align*}
where each $x^i$ is inductively defined as follows, with $1 \leq \ell < k$:
\begin{itemize}

\item $x^1 = \Plr(\P)$ with some tie-breaking;

\item $x^{\ell + 1} = \Plr(\P|_{\A \setminus \set{x^1, \ldots, x^\ell}})$ with some tie-breaking.

\end{itemize}
\end{VR}

It is worth observing that sequential plurality and $k$-plurality are different rules as they elect different committees (see Exercise \ref{ex:cplurality}).

\begin{VR}[Single transferable vote for committees] \label{vr:CSTV}
The single transferable vote rule for committee selection (CSTV) is the CSF defined as follows, for any profile $\P$ and $k \leq 1$:
\[
\CSTV_k(\P) = \set{X \mid \exists S^\ell = \tuple{X^\ell, \P^\ell} \ \mbox{s.t.} \ \size{X^\ell} = k }
\]
where stages $S^\ell$ are recursively defined as follows:
\begin{align*}
S^0 & = \tuple{X^0, \P^0} \ \mbox{where} \ X^0 = \emptyset, \P^0 = \P \\
S^{\ell + 1} & =
\left\{
\begin{array}{ll}
\tuple{X^\ell \cup \set{x}, \left(\P^\ell|_{\A \setminus \set{x}}\right)_{-C}}  & \mbox{if} \ \Plr(\P^\ell)(x) \geq q \\
\tuple{X^\ell, \P^\ell|_{\A \setminus \set{y}}}  & \mbox{otherwise}
\end{array}
\right.
\end{align*}
where: 
$\Plr(\P^\ell)$ denotes the plurality score;
$y \in \argmin_{z \in \A}(\Plr(\P^\ell)(z))$, i.e., $y$ has lowest plurality score;
$C \subseteq N$ is a set of size $q$ consisting of voters who rank $x$ as top in their preferences in profile $\P^\ell$; and $q = \lfloor \frac{n}{k+1} \rfloor + 1$.
\end{VR}
Intuitively at each stage the following happens: an alternative is added to the committee if the plurality score of such alternative meets the threshold $q$ (there may be many such sets), and if that is the case the alternative is discarded as well as a set of voters of size $q$ (there may be many such sets); otherwise no alternative is added to the committee and a plurality loser (there may be many such alternatives) is discared. The above is repeated until a committee of size $k$ is constructed.

Observe that the rule is non-deterministic, involving several tie-breaking decisions at each step. Intuitively all different committees resulting from different tie-breaking decisions are recorded in the output of the rule.\footnote{This is sometimes referred to as the parallel-universe tie-breaking model.} 

\begin{example}[\cite{faliszewski17multiwinner}]
Consider the following profile for $\set{1, \ldots, 6}$ and $\set{a, b, c, d, e}$:
\[
\P = 
\begin{array}{l | lllll}
1 & a & b & c & d & e \\
2 & e & a & b & d & c \\
3 & d & a & b & c & e \\
4 & c & b & d & e & a \\
5 & c & b & e & a & d \\
6 & b & c & d & e & a \\
\end{array} 
\]
We have:
\begin{align*}
\CSTV_2(\P) & = \set{\set{b,c}, \ldots \mbox{left to the reader} \ldots} \\
\SnC^\Plr_2(\P) & = \set{\set{c,a}, \set{c,b}, \set{c,d}, \set{c,e}} \\
\ChCo^{\tuple{m-1, m-2, \ldots, 0}}_2(\P) & = \set{\set{a, c}} \\
\PAV_2(\P) & = \set{\set{a,b}}
\end{align*}
In the $\ChCo^{\tuple{m-1, m-2, \ldots, 0}}_2$ rule $a$ represents voters in $\set{1, 2, 3}$ and $c$ represents voters in $\set{4, 5, 6}$. In the $\PAV_2$ the winning committee $\set{a,b}$ obtains $6.5$ points.
\end{example}

\begin{remark}[Droop quota]
The quota $q = \lfloor \frac{n}{k+1} \rfloor + 1$ in the committee STV Rule \ref{vr:CSTV} above is known as the Droop quota, from Richmond Droop, English lawyer and mathematician who introduced it (see \cite{droop81methods}). It is the most commonly used quota for STV committee elections, for instance in the Republic of Ireland. The quota is the smallest integer that guarantees that no candidate who would reach the quota would then have no place available in the committee. So it is a generalization of the idea of simple majority when electing committees of size $k = 1$. It can be derived as follows. 
The quota $q$ should satisfy two constraints: $k \cdot q \leq n$, i.e., the number of candidates meeting the quota cannot exceed the number of voters; and $n \leq k \cdot q + (q - 1)$, i.e., there cannot be a  $k+1^\mathit{th}$ candidate who meets the quota. From this we obtain:
\[
\left\lceil \frac{n+1}{k+1}\right\rceil = \left\lfloor \frac{n}{k+1} \right\rfloor + 1 \leq q \leq \frac{n}{k}.
\]
Taking the smallest $q$ satisfying the above inequalities thus gives us the Droop quota.
\end{remark}

%%%%%%%%%%%%%%%%%%%%%%%%%%%%

\section{Axiomatic results} \label{sec:csf_axioms}

In this section we showcase the application of the axiomatic method to the case of multi-winner rules. We focus on the generalization of the notion of Condorcet-consistency (recall Definition \ref{def:more}) to the case of committee selection.

\subsection{Condorcet consistency and monotonicity}

\begin{definition}[Condorcet committees \cite{gehrlein85condorcet}]
Let a profile $\P$ (for $\tuple{\N, \A}$) be given. A set $C \subseteq \A$ is a {\em weak Condorcet committee} if, for all $x \in C$ and $y \not\in C$, $\supp^{xy}_\P \geq \supp^{yx}_\P$. For a given profile $\P$, we denote its set of Condorcet commitees of size $k$ by $\CC_k(\P)$.
\end{definition}
Intuitively, a committee is Condorcet consistent whenever it does not elect any alternative for which a different alternative is preferred by a majority of voters. Notice that the set of Condorcet committees for a given size $k$ may be empty.

\begin{example}[\cite{barbera08how}] \label{ex:imp}
\[
\P = 
\begin{array}{l | lllll}
1 & a & b & c & d & e \\
2 & a & b & e & c & d \\
3 & a & b & d & e & c \\
4 & c & d & e & a & b \\
5 & e & c & d & a & b \\
6 & d & e & c & a & b \\
\end{array} 
\]
We have that $\CC_1(\P) = \set{\set{a}}$, $\CC_2(\P) = \set{\set{a,b}}$ and $\CC_3(\P) = \set{\set{c, d, e}}$.

\end{example}

\begin{definition}
A family of CSFs $f$ (for a given $\tuple{\N, \A}$) is:

\begin{description}

\item[Condorcet consistent] (or) {\bf stable} iff for all $\P \in \LO(\A)^n$, and $1 \leq k \leq n$, $f_k(\P) \subseteq \CC_k(\P)$ whenever $\CC_k(\P) \neq \emptyset$. 

{\em Intuitively}, the rule always selects a weak Condorcet committee when one exists.

\item[Committee monotonic] iff for all $\P \in \LO(\A)^n$ and $1 \leq k < n$, 
\begin{itemize}
\item if there exists $C \in f_k(\P)$ then there exists $C' \in f_{k+1}(\P)$ s.t. $C \subset C'$;
\item if there exists $C' \in f_{k+1}(\P)$ then there exists $C \in f_{k}(\P)$ s.t. $C \subset C'$.
\end{itemize}

{\em Intuitively}, an alternative selected for a small committee should be selected also in larger committees.

\end{description}

\end{definition}
Observe that the first condition of committee monotonicity would suffice to express the axiom if $f$ is a family of resolute CSFs.

\begin{theorem}[\cite{barbera08how}]
There exists no family of CSFs $f$ such that $f$ is Condorcet consistent and committee monotonic.
\end{theorem}
\begin{proof}
Assume towards a contradiction that this is not the case: there exists $f_k$ such that, for any $k$, $f_k$ is Condorcet consistent and committee monotonic.
Let $\P$ be the profile of Example \ref{ex:imp}. Since $f_2$ is Condorcet consistent by assumption, $\set{a, b} \in f_2(\P)$. Similarly, $\set{c, d, e} \in f_3(\P)$. This is a failure of committee monotonicity. Contradiction.
\end{proof}

\begin{theorem}[\cite{elkind17properties}]
A family of CSFs $f$ is committee-monotonic if and only if it is a best-$k$ rule.
\end{theorem}

\begin{proof}[Sketch of proof]
\RtoL Best-$k$ rules are committee monotonic by construction (Definition \ref{def:SnC}). \LtoR The proof is by construction. Given a committee monotonic multiwinner rule $f$ we can define a SPF $G^f$ which, for every profile $\P$, chooses 
a social preference (total pre-order) $\preceq$ such that the alternatives in $\preceq$ with rank $k$ are precisely the elements in $f_{k}(\P) \backslash f_{k-1}(\P)$. Given $G$ Definition \ref{def:SnC} gives us a best-$k$ rule as desired.
\end{proof}

\subsection{Proportionality}

In recent years, the application of the axiomatic method to CSFs has  focused especially on trying to tease out the many ways in which the notion of proportionality can actually be interpreted. A wealth of axioms---as well as rules inspired by them---have been studied which capture different notions of proportionality, and highlight the richness of the concept. For example: extended justified representation \cite{aziz17justified}; laminar proportionality \cite{peters20proportionality}; priceability \cite{peters20proportionality}.

%%%%%%%%%%%%%%%%%%%%%%%%%%%%%

\section{Chapter notes}

Multiwinner rules are a very recent, and rapidly developing, area of research within computational social choice.
The chapter is based on \cite{barbera08how,elkind17properties,faliszewski17multiwinner}.

%%%%%%%%%%%%%%%%%%%%%%%%%%%%%

\section{Exercises}

%\begin{exercise} \label{ex:cplurality}
%Construct a profile $\P$, and fix a $k$, such that $\SPlr_k(\P) \neq \SnC^\Plr_k(\P)$. Try to make such profile as small as possible.
%\end{exercise}

\begin{exercise}  \label{ex:cplurality}
Construct a profile, and fix a $k$, such that $\SPlr_k(\P) $ (Rule \ref{vr:splr}) and $\SnC^\Plr_k(\P)$ (one of the rules in Rule \ref{vr:SnC}) output different committees of size $k$ for $\P$. Try to make such profile as small as possible. Based on your example explain the difference between the two rules.
\end{exercise}

\begin{exercise}
Construct a profile, and fix a $k$, such that $\SnC^{\Ar_k}_k$ (one of the rules in Rule \ref{vr:SnC}) and $\PAV_k$ (Rule \ref{vr:PAV}) output different committees of size $k$ for $\P$. Try to make such profile as small as possible. Based on your example explain the difference between the two rules.
%Explain the difference between $\CSTV_1$ (Rule \ref{vr:CSTV}) and $\STV$ (Rule \ref{vr:STV}).
\end{exercise}

\begin{exercise}
Every committee selection function $f_k$ defines a social choice function when $k = 1$. What are the social choice functions defined by:
\begin{enumerate}[i)]
\item $\ChCo_1^\w$, where $\w$ is the Borda scoring vector (Definition \ref{vr:chco});
\item $\ChCo_1^\w$, where $\w$ is the $\ell$-approval vector, where $1 \leq \ell \leq m$ (recall Definition \ref{vr:approval});
\item $\PAV_1$ (Definition \ref{vr:PAV}).
\end{enumerate}
Explain your answers.
\end{exercise}

%%%%%%%%%%%%%%%%%%%%%%%%%%%%%%%%%%%%%%%%%%%%%%%%%%%%%%%%%%%%%%%

\chapter{Topics for Projects}

%%%%%%%%%%%%%%%%%%%%%%

This is a list of topics, in no specific order and with relevant blibliography, to be chosen from for the final papers.

\paragraph{Liquid democracy}

Liquid democracy is a form of voting where voting rights can be delegated transitively. It has been advocated in \cite{liquid_feedback} and used by, among others, the Pirate Party in Germany and France.
A series of recent papers have focused on various aspects of the system: \cite{kling15voting,blum16liquid,christoff17binary,bloembergen18rational,kahng18liquid,caragiannis19contribution,zhang21power}. See also \cite{behrens17origins} for a history of the concept.

\paragraph{Participatory budgeting}

The participatory budgeting problem consists in selecting a set of projects within a given budget, based on the preferences of citizens in a given area (e.g., municipality)\footnote{Maybe good to know that a participatory budgeting pilot has been run in autumn 2019 Oosterparkwijk, Groningen. One more is currently being run in Helpman, Groningen}. A number of papers have recently proposed and studied different voting rules for participatory budgeting and studied their properties: \cite{benade17preference,aziz18proportionally,goel19knapsack,jain20participatory}. Again the aim is to identify `optimal' rules. See \cite{aziz20participatory} for a recent overview.

\paragraph{Voting with preference intensity}

The type of voting covered in these lecture notes adhered to the `one-voter-one-vote' principle. This makes it impossible for voters to signal the `intensity' of their preferences (how much I like $x$ over $y$). Voting rules have been suggested that make the signaling of preference intensities possible. Here are some examples: cumulative voting \cite{glasser59game}, storable voting \cite{casella05storable,casella06experimental}, quadratic voting \cite{lalley18quadratic}.

\paragraph{Evaluative voting}

We are accustomed to an idea of voting in which we are asked to express our preferences, rather than judging on the quality of a candidate or a policy proposal. A number of new voting methods have been proposed which more clearly frame the social choice problem as a problem of evaluation: approval \& disapproval voting \cite{alcantud14disapproval}, majority judgments \cite{balinski14judge}.

\paragraph{Randomized voting}

In the early days of democracy randomization played an important role and was considered a quintessentially fair mechanism (e.g., in the Athenian democracy public officials were selected by lottery). Growing research is currently focusing on the use of randomization to improve on standard deterministic voting-based social choice. A good recent overview article is \cite{brandt18rolling}.

\paragraph{Sybil-resilient voting}

A key problem of voting mechanisms over the internet is the possibility of Sybil attacks (that is, voters can generate arbitrarily many identities thereby manipulating the outcome). Considerable research has been dedicated to the development of voting mechanisms that can be, to a smaller or larger extent, resistant to this form of attack \cite{conitzer10using,todo11false,waggoner12evaluating,wagman08optimal,wagman14false,shahaf19sybil}.

\paragraph{Epistemic social choice}

Epistemic social choice develops the insights of jury theorems such as Condorcet's by making more realistic assumptions over voters competence and independence. A good recent overview article is \cite{pivato19realizing}.

\paragraph{Judgment aggregation}

We have studied the social choice problem as a problem of aggregation of preferences, but more generally it can be framed as a problem of aggregation of logical formulas. This is the perspective taken in so-called judgment aggregation \cite{endriss16judgment,Grossi_2014}.

\paragraph{Strategic voting}

In Chapter \ref{ch:truthful} we discussed the issue of strategic or tactical voting, and touched on some basic game-theoretic aspects. Extensive literature exists on strategic voting and a comprehensive overview of the topic is \cite{meir18strategic}.

\paragraph{Doodle pool voting}

Doodle pools work with a form of approval voting (cf. Rule \ref{vr:approval}) giving rise to interesting forms of manipulative behavior. A few recent papers have looked into this issue \cite{james15strategic,obraztsova17doodle,anthony18how} and modeled Doodle pools as a special type of voting games.

\paragraph{Agenda manipulation}

SCFs are executed by central authorities (government, businesses, etc.). Such central authorities running the voting mechanisms may have the power to decide the agenda (the alternatives) of a social choice context. Can they, by changing such set, obtain better outcomes for themselves? This is the issue of agenda setting or agenda manipulation. An introduction to the topic, with further relevant references, is \cite[Section 2.4]{taylor05social}.

\paragraph{Incomplete preferences and elicitation}

In this course we assumed profiles consisted of complete descriptions of the preferences of all agents. However, in many settings, eliciting the full linear order from each agent may be unfeasible (e.g., because of the sheer size of the set of alternatives). Recent work has extended some of the notions we studied to the setting with incomplete preferences: \cite{terzopoulou19aggregating,kruger20strategic}. A related line of research has looked at the problem of how to best elicit voter's preferences in order to determine winners, depending on the different voting rules: how many queries, and how complex, does a voting rule need in order to compute a social choice? A good starting point for this line of work is \cite[Ch. 10]{comsoc_handbook}. The problem is also directly relevant to applications in participatory democracy (see, e.g. \cite{lee14crowdsourcing}).

\paragraph{Participation axioms}
An SCF suffers of the {\em no show paradox} whenever there exist profiles in which an agent would do better by not casting their ballot under the rule. An SCF is then said to satisfy the participation axiom if it does not suffer of the no show paradox. A good introduction to this axiom, how it relates to monotonicity and strategy-proofness axioms, as well as relevant related literature, can be found in \cite[Ch. 2]{comsoc_handbook}.

\paragraph{Voting theory and ensemble classifiers}
An established approach in Machine Learning to improve the performance in classification tasks is by `combining' the decisions of various individual classifiers. One way to combine classifiers is to use voting and some papers have explored how different voting rule affect performance in classification tasks by ensembles \cite{mu09analysis,cornelio19voting}.

%%%%%%%%%%%%%%%%%%%%%%%%%%%%%%%%%%%%%%%%%%%%%%%%%%%%%%%

% Bibliography
\newpage

\bibliographystyle{apalike}
%\bibliography{COMSOC_bib.bib}

\end{document}